\pdfoutput=1
\documentclass[10pt, a4paper]{article}

\usepackage[utf8]{inputenc}
\usepackage[T1]{fontenc}
\usepackage{lmodern} 
\usepackage{fontawesome}

\usepackage[margin=2cm]{geometry}

\usepackage{fancyhdr}
\usepackage[us,24hr]{datetime} 
\fancypagestyle{plain}{
\fancyhf{}
\rfoot{}
\lfoot{Page \thepage}
}
\usepackage[activate={true,nocompatibility},final,tracking=true,
kerning=true,spacing=true,factor=1100,stretch=10,shrink=10]{microtype}
\SetTracking{encoding={*}, shape=sc}{40}
\usepackage{ellipsis}         

\usepackage[abbreviations,british]{foreign} 

\usepackage{scalerel}

\usepackage{etoolbox}

\usepackage{xcolor}
\usepackage{hyperref}

\definecolor{darkmidnightblue}{rgb}{0.0, 0.2, 0.4}
\definecolor{persianplum}{rgb}{0.44, 0.11, 0.11}

\hypersetup{
  colorlinks  = true, 
  citecolor   = persianplum, 
  urlcolor    = persianplum, 
  linkcolor   = persianplum
}

\usepackage{tikz}
\usepackage{todonotes}
\usepackage{xspace}

\usepackage{amssymb} 
\usepackage{amsmath}
\usepackage{amsthm}
\usepackage{mathtools}
\usepackage[bb=boondox]{mathalfa} 

\usepackage[capitalise, noabbrev]{cleveref}

\newtheorem{thm}{Theorem}[section]
\newtheorem*{thm*}{Theorem}
\newtheorem{definition}[thm]{Definition}

\newtheorem{corollary}[thm]{Corollary}
\newtheorem{example}[thm]{Example}
\newtheorem{fact}[thm]{Fact}
\newtheorem{lemma}[thm]{Lemma}

\usetikzlibrary{arrows, decorations.markings, shapes, calc}
\usepackage{graphics/nicolas-markey-expose}

\usetikzlibrary{shadows}
\tikzset{
diagonal fill/.style 2 args={fill=#2, path picture={
\fill[#1, sharp corners] (path picture bounding box.south west) -|
                         (path picture bounding box.north east) -- cycle;}},
reversed diagonal fill/.style 2 args={fill=#2, path picture={
\fill[#1, sharp corners] (path picture bounding box.north west) |- 
                         (path picture bounding box.south east) -- cycle;}}
}

\usetikzlibrary{hobby,backgrounds,calc,trees}

\usepackage{tikz-cd}
\usetikzlibrary{matrix}

 \makeatletter
 \def\desclabel#1#2{\begingroup
    \def\@currentlabel{#1}%
    #1\label{#2}\endgroup
 }
 \makeatother

\usepackage{caption}
\usepackage{subcaption}
\usepackage{float}

\usepackage{braket}

\usepackage[shortlabels]{enumitem} 
\usepackage{multicol}

\usepackage[h]{esvect}

\usepackage[linesnumbered,ruled,vlined]{algorithm2e}

\SetCommentSty{mycommfont}
\SetKwInput{KwData}{Input}

\usepackage{xurl} 

\newcommand{\orcid}[1]{\href{https://orcid.org/#1}{\includegraphics[width=9pt]{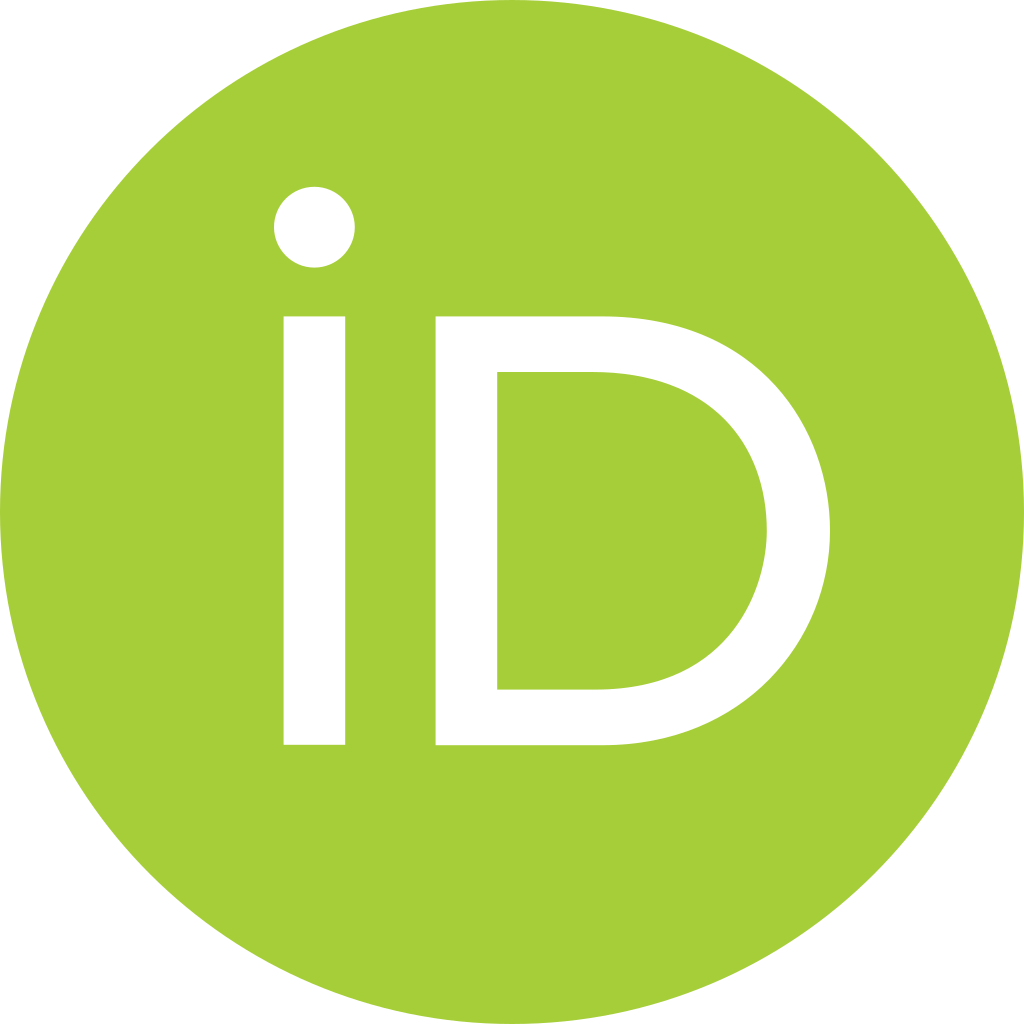}}}

 \makeatletter
 \def\desclabel#1#2{\begingroup
    \def\@currentlabel{#1}%
    #1\label{#2}\endgroup
 }
 \makeatother

\definecolor{ao(english)}{rgb}{0.0, 0.5, 0.0}

\newrobustcmd{\newnotion}[1]{\emph{#1}}

\newcommand{\dland}{\sqcap}
\newcommand{\dlor}{\sqcup}
\newcommand{\topconcept}{\top}
\newcommand{\botconcept}{\bot}
\newcommand{\dlsubseteq}{\sqsubseteq}

\makeatletter
\providecommand{\bigdland}{%
  \mathop{%
    \mathpalette\@updown\bigsqcup
  }%
}
\newcommand*{\@updown}[2]{%
  \rotatebox[origin=c]{180}{$\m@th#1#2$}%
}
\makeatother

\newcommand{\DL}[1]{\ensuremath{\mathcal{#1}}}  

\newcommand{\ALC}{\DL{ALC}}                     
\newcommand{\ALCI}{\DL{ALCI}}                   
\newcommand{\ALCO}{\DL{ALCO}}                   
\newcommand{\ALCQ}{\DL{ALCQ}}                   
\newcommand{\ALCH}{\DL{ALCH}}                   
\newcommand{\ALCSelf}{\DL{ALC}_{\textsf{Self}}} 
\newcommand{\ALCHQ}{\DL{ALCHQ}}                 
\newcommand{\ALCSCC}{\DL{ALCSCC}}               
\newcommand{\ALCOIQ}{\DL{ALCOIQ}}               
\newcommand{\ALCcap}{\ALC^{\cap}}               
\newcommand{\muALC}{\mu\ALC}                    
\newcommand{\ALCHbregQ}{\DL{ALCH}\textit{b}_{\textsf{reg}}\DL{Q}}

\newcommand{\SHQ}{\DL{SHQ}}                 

\newcommand{\Z}{\DL{Z}}                         

\newcommand{\abstrDL}{\DL{L}}     

\newcommand{\GF}{\DL{GF}}

\newcommand{\complexityclass}[1]{\textsc{#1}} 
\newcommand{\ExpTime}{\complexityclass{ExpTime}} 
\newcommand{\TwoExpTime}{\complexityclass{2ExpTime}} 
\newcommand{\lang}[1]{\mathbf{#1}}  
\newcommand{\Ilang}{\lang{N_I}}     
\newcommand{\Rlang}{\lang{N_R}}     
\newcommand{\Clang}{\lang{N_C}}     
\newcommand{\Vlang}{\lang{N_V}}     

\newcommand{\query}[1]{\mathit{#1}}  
\newcommand{\queryq}{\query{q}}      
\newcommand{\match}[1]{#1}          
\newcommand{\matchpi}{\match{\pi}}  
\newcommand{\modelsmatch}[1]{\models_{#1}} 
\newcommand{\modelsfin}{\models_\mathrm{fin}}       
\newcommand{\modelsoptfin}{\models_\mathrm{(fin)}}  
\newcommand{\queryVar}[1]{\mathrm{Var}{(#1)}}   
\newcommand{\queryVarq}{\queryVar{\queryq}}     

\newcommand{\var}[1]{\mathit{#1}}   
\newcommand{\varx}{\var{x}}         
\newcommand{\vary}{\var{y}}         
\newcommand{\varz}{\var{z}}         
\newcommand{\varv}{\var{v}}         
\newcommand{\varu}{\var{u}}         
\newcommand{\varw}{\var{w}}         

\newcommand{\role}[1]{\mathit{#1}}      
\newcommand{\roler}{\role{r}}           
\newcommand{\roles}{\role{s}}           

\newcommand{\concepts}{\lang{C}}            
\newcommand{\alccapconcepts}{\concepts} 

\newcommand{\concept}[1]{\mathrm{#1}}       
\newcommand{\conceptA}{\concept{A}}         
\newcommand{\conceptB}{\concept{B}}         
\newcommand{\conceptC}{\concept{C}}         
\newcommand{\conceptD}{\concept{D}}         

\newcommand{\indv}[1]{\texttt{#1}}  
\newcommand{\indva}{\indv{a}}       
\newcommand{\indvb}{\indv{b}}       
\newcommand{\indvc}{\indv{c}}       

\newcommand{\kb}[1]{\mathcal{#1}}   
\newcommand{\kbK}{\kb{K}}           
\newcommand{\abox}[1]{\mathcal{#1}} 
\newcommand{\aboxA}{\abox{A}}       
\newcommand{\tbox}[1]{\mathcal{#1}} 
\newcommand{\tboxT}{\tbox{T}}       
\newcommand{\ind}[1]{\mathsf{ind}{(#1)}} 
\newcommand{\indA}{\ind{\aboxA}}    
\newcommand{\indK}{\ind{\kbK}}    
\newcommand{\names}{\mathsf{N}}     

\newcommand{\inter}[1]{\mathcal{#1}}    
\newcommand{\interI}{\inter{I}}         
\newcommand{\interJ}{\inter{J}}         

\newcommand{\DeltaInter}[1]{\Delta^{#1}}                
\newcommand{\DeltaI}{\DeltaInter{\interI}}              
\newcommand{\DeltaJ}{\DeltaInter{\interJ}}              

\newcommand{\cdotInter}[1]{\cdot^{#1}}          
\newcommand{\cdotI}{\cdotInter{\interI}}        

\newcommand{\interIquery}[1]{\inter{I}_{#1}}        
\newcommand{\interIqueryq}{\interIquery{\queryq}}   
\newcommand{\cdotIqueryq}{\cdotInter{\interIqueryq}}   

\newcommand{\domelem}[1]{\mathrm{#1}}                           
\newcommand{\domelemc}{\domelem{c}}                             
\newcommand{\domelemd}{\domelem{d}}                             
\newcommand{\domeleme}{\domelem{e}}                             

\newcommand{\Deltanamed}[2]{\DeltaInter{#1}_{\mathrm{named}{(#2)}}}   
\newcommand{\DeltaInamednames}{\Deltanamed{\interI}{\names}}        

\newcommand{\deff}{:=}

\newcommand{\N}{\mathbb{N}}
\renewcommand{\set}[1]{\{#1\}}

\newcommand{\pathrho}{\rho}

\renewcommand{\restriction}{\mathord{\upharpoonright}}
\newcommand{\restr}[2]{#1\restriction_{#2}} 
\newcommand{\Nbhood}[3]{\mathsf{Nbd}^{#2}_{#1}{(#3)}}  

\newcommand{\homo}[1]{\mathfrak{#1}}    
\newcommand{\homof}{\homo{f}}           
\newcommand{\homoh}{\homo{h}}           

\newcommand{\datadomD}{\mathbb{D}}          
\newcommand{\isancestor}{\prec}             
\newcommand{\isancestoreq}{\preceq}         
\newcommand{\subtreeofrootedat}[2]{#1^{[#2 \isancestoreq]}} 
\newcommand{\childrenof}[2]{\mathsf{Chlds}_{#1}{(#2)}}       
\newcommand{\ancestorsof}[2]{\mathsf{Ancrs}_{#1}{(#2)}}      
\newcommand{\Childrenofvar}[1]{\mathsf{Chlds}{(#1)}}       

\newcommand{\interfwdunrav}[3]{{#1}^{#2}_{#3}}  
\newcommand{\Deltafwdunrav}[3]{\DeltaInter{\interfwdunrav{#1}{#2}{#3}}}  
\newcommand{\cdotfwdunrav}[3]{\cdotInter{\interfwdunrav{#1}{#2}{#3}}}    

\newrobustcmd{\interomegafwdunravI}[1]{\interfwdunrav{\interI}{\protect\vv{\omega}}{#1}}   
\newrobustcmd{\interomegafwdunravInames}{\interomegafwdunravI{\names}}  
\newrobustcmd{\interomegafwdunravIindA}{\interomegafwdunravI{\indA}}

\newrobustcmd{\DeltaomegafwdunravI}[1]{\Deltafwdunrav{\interI}{\protect\vv{\omega}}{#1}}   
\newrobustcmd{\DeltaomegafwdunravInames}{\DeltaomegafwdunravI{\names}}  
\newrobustcmd{\DeltaomegafwdunravIindA}{\DeltaomegafwdunravI{\indA}}

\newrobustcmd{\cdotomegafwdunravI}[1]{\cdotfwdunrav{\interI}{\protect\vv{\omega}}{#1}}     
\newrobustcmd{\cdotomegafwdunravInames}{\cdotomegafwdunravI{\names}}  
\newrobustcmd{\cdotomegafwdunravIindA}{\cdotomegafwdunravI{\indA}}

\newcommand{\splittingof}[2]{\Pi_{#1}^{#2}}                 
\newcommand{\splittingofq}[1]{\splittingof{\queryq}{#1}}       
\newcommand{\splittingtrees}{\mathtt{Trees}}                    
\newcommand{\splittingroots}{\mathtt{Roots}}                    
\newcommand{\splittingithsubtree}[1]{\mathtt{SubTree}_{#1}}     
\newcommand{\splittingname}{\mathtt{name}}                      
\newcommand{\splittingrootof}{\mathtt{root}\text{-}\mathtt{of}} 

\newcommand{\spoil}{{\scaleobj{0.75}{\text{\faBolt}}}}
\newcommand{\superspoil}{{\spoil}^{\scaleobj{0.4}{\text{\faStar}}}}
\newcommand{\spoilerKBof}[1]{\kbK_{#1}^{\spoil}}

\newcommand{\superspoilerKBof}[1]{\kbK_{#1}^{\superspoil}}

\newcommand{\subtreeConcept}[2]{\concept{Subt}_{#1}^{#2}}  
\newcommand{\matchConcept}[1]{\concept{Match}_{#1}}        

\newcommand{\eliminateforkto}{\leadsto_{\mathsf{fe}}}        
\newcommand{\maximalforkrew}[1]{\mathsf{maxfr}{(#1)}}  
\newcommand{\QTree}{\mathsf{QTree}}
\newcommand{\Reach}{\mathsf{Reach}}
\newcommand{\SAT}{\mathsf{SAT}}

\title{Lutz's Spoiler Technique Revisited: A Unified Approach to Worst-Case Optimal Entailment of Unions of Conjunctive Queries\\ in Locally-Forward Description Logics}

\author{Bartosz Bednarczyk$^{1,2}$ \orcid{0000-0002-8267-7554}}
\date{$^{1}$Computational Logic Group, Technische Universit{\"a}t Dresden, Germany\\
$^{2}$Institute of Computer Science, University of Wroc{\l}aw, Poland\\
\texttt{bartosz.bednarczyk@$\{$cs.uni.wroc.pl, tu-dresden.de$\}$}}

\begin{document}
\maketitle


\begin{abstract}
We present a unified approach to (both finite and unrestricted) worst-case optimal entailment of (unions of) conjunctive queries (U)CQs in the wide class of ``locally-forward'' description logics. 
The main technique that we employ is a generalisation of Lutz's spoiler technique, originally developed for CQ entailment in $\ALCHQ$.
Our result closes numerous gaps present in the literature, most notably implying $\ExpTime$-completeness of UCQ-querying for any superlogic of $\ALC$ contained in $\ALCHbregQ$, and, as we believe, is abstract enough to be employed as a black-box in many new scenarios.
\end{abstract}

\section{Preliminaries}\label{sec:prelim}
  We recall the basics on description logics (DLs)~\cite{dlbook} and query answering~\cite{OrtizS12}.

\paragraph*{DLs.}\label{subsec:prelim-dls}
  We fix countably infinite pairwise disjoint sets of \emph{individual names} \(\Ilang \), \emph{concept names} \(\Clang \), and \emph{role names}~\(\Rlang \) and introduce a description logic \( \ALCcap \).
  Starting from \( \Clang \) and \( \Rlang \), the set \( \alccapconcepts \) of \( \ALCcap \) \emph{concepts} is built using the following concept constructors: \emph{negation} \((\neg \conceptC) \), \emph{conjunction} \((\conceptC \dland \conceptD) \), \emph{existential restriction} (\(\exists{(\roler_1 \cap \ldots \cap \roler_n)}.\conceptC \)) and the \emph{bottom concept} (\( \botconcept \)), with the grammar:
  \begin{equation*} \label{eq:alc-grammar}
  \conceptC, \conceptD \; ::= \; \botconcept \; \mid \; \conceptA \; \mid \; \neg \conceptC \; \mid \; \conceptC \dland \conceptD \; \mid \; \exists{(\roler_1 \cap \ldots \cap \roler_n)}.\conceptC,
  \end{equation*}
  where \(\conceptC,\conceptD \in \alccapconcepts \), \(\conceptA \in \Clang \) and \(\roler \in \Rlang \). 
  We often employ disjunction \(\conceptC \dlor \conceptD \deff \neg (\neg \conceptC  \dland \neg \conceptD) \), universal restrictions \(\forall{(\roler_1 \cap \ldots \cap \roler_n)}.\conceptC \deff \neg \exists{(\roler_1 \cap \ldots \cap \roler_n)}.\neg\conceptC \), top \(\topconcept \deff \neg \botconcept \), and the ``inline-implication''~$\conceptC \to \conceptD \deff \neg \conceptC \dlor \conceptD$.

  \emph{Assertions} are of the form \(\conceptC(\indva) \) or \(\roler(\indva,\indvb) \) for \(\indva,\indvb \in \Ilang \), \(\conceptC \in \alccapconcepts \) and \(\roler \in \Rlang \).
  A \emph{general concept inclusion} (GCI) has the form \(\conceptC \dlsubseteq \conceptD \) for concepts \(\conceptC, \conceptD \in \alccapconcepts \). 
  We use $\conceptC \equiv \conceptD$ as a shorthand for the two GCIs \(\conceptC \dlsubseteq \conceptD \) and \(\conceptD \dlsubseteq \conceptC \).
  A~\emph{knowledge base} (KB) \(\kbK=(\aboxA, \tboxT) \) is composed of a finite non-empty set \(\aboxA \) (\emph{ABox}) of assertions and a finite non-empty set \(\tboxT \) (\emph{TBox}) of GCIs. 
  We call the elements of \(\aboxA \cup \tboxT \) \emph{axioms}. 
  The set of all individual names appearing in \( \kbK \) is denoted with \( \indK \). 

  \begin{table}[!htb]
    \begin{minipage}{.65\linewidth}
      \caption{Concepts and roles in \(\ALCcap \).\label{tab:ALCcap}}
      \centering
          \begin{tabular}{@{}l@{\ \ \ }c@{\ \ \ }l@{}}
              \hline\\[-2ex]
              Name & Syntax & Semantics \\ \hline \\[-2ex]
              bottom concept & \( \botconcept \) & \( \emptyset  \) \\
              conc.\ negation & \(\neg\conceptC \)& \(\DeltaI \setminus \conceptC^{\interI} \) \\  
              conc.\ intersection & \(\conceptC \dland \conceptD \)& \(\conceptC^{\interI}\cap \conceptD^{\interI} \) \\  
              exist.\ restriction & \(\exists{(\roler_1 \cap \ldots \cap \roler_n)}.\conceptC \) & 
              \(\set{ \domelemd \mid \exists{\domeleme}.(\domelemd,\domeleme)\in \bigcap\limits_{i = 1}^{n} \roler_i^{\interI} \land \domeleme\in \conceptC^{\interI}} \)
              \\\hline
          \end{tabular}
    \end{minipage}%
    \begin{minipage}{.35\linewidth}
      \caption{Axioms in \( \ALCcap \).\label{tab:axm}} 
      \centering
          \begin{tabular}{ l l }
              \hline\\[-2ex]
              Axiom \(\alpha \) & \(\interI \models\alpha \), if \\ \hline \\[-2ex] 
              \(\conceptC \dlsubseteq \conceptD \) & \(\conceptC^{\interI} \subseteq \conceptD^{\interI}  \)\hspace{5ex} \mbox{TBox}~\(\tboxT \) \\ \hline \\[-2ex]
              \(\conceptC(\indva) \) & \(\indva^{\interI} \in \conceptC^{\interI}  \)\hfill \mbox{ABox}\(~\aboxA \) \\
              \(\roler(\indva,\indvb) \) & \((\indva^{\interI}, \indvb^{\interI} )\in \roler^{\interI}  \)\\
              \(\neg\roler(\indva,\indvb) \) & \((\indva^{\interI}, \indvb^{\interI} )\not\in \roler^{\interI}  \)             
              \\\hline
          \end{tabular}
    \end{minipage} 
  \end{table}
  The semantics of \(\ALCcap \) is defined via \emph{interpretations} \(\interI = (\DeltaI, \cdotI) \) composed of a non-empty set \(\DeltaI \) called the \emph{domain of \(\interI \)} and an \emph{interpretation function} \(\cdotI \) mapping individual names to elements of \(\DeltaI \), concept names to subsets of \(\DeltaI \), and role names to subsets of \(\DeltaI \times \DeltaI \). 
  This mapping is extended to concepts (see~\cref{tab:ALCcap}) and finally used to define \emph{satisfaction} of assertions and GCIs (see~\cref{tab:axm}). 
  \emph{Structures} are interpretations with a partial assignment of individual names.
  We say that an interpretation \(\interI \) \emph{satisfies} a KB \(\kbK=(\aboxA,\tboxT) \) (or \(\interI \) is a \emph{model} of \(\kbK \), written: \(\interI \models \kbK \)) if it satisfies all axioms of~\(\aboxA\cup\tboxT \). 
  An interpretation $\interI$ is \emph{finite} (resp. countable) iff its domain $\DeltaI$ is finite (resp. countable).
  A KB is (finitely) \emph{consistent} (or (finitely) \emph{satisfiable}) if it has a (finite) model and (finitely) \emph{inconsistent} (or (finitely) \emph{unsatisfiable}) otherwise. 

  Given a set of individual names \(\names \subseteq \Ilang \) we denote with \( \DeltaInamednames \) the set of \emph{\(\names \)-named domain elements} of~\(\interI \), \ie{} the set of all \(\domelemd \in \DeltaI \) for which \(\domelemd = \indva^{\interI} \) holds for some name~\(\indva \in \names \). 
  The elements from its complement, namely from \(\DeltaI \setminus \DeltaInamednames \), are called \(\names \)-\emph{anonymous}.

  The presented notions are straightforwardly lifted to any description logic \( \abstrDL \) semantically extending \( \ALCcap \) and allowing for polynomial expressivity of $\ALCcap$ concepts.
  Throughout the paper, such logics will be called \newnotion{abstract expressive description logics} or simply \newnotion{abstract~DLs}.\footnote{
  We have decided not to formally define what a \emph{semantic extension} of $\ALCcap$ is, suggesting that this notion should rather be understood naively.
  Promising examples of \emph{abstract DLs} are well-known DLs like \(\ALC^{\cap}, \ALCOIQ^{\cap}, \SHQ^{\cap}, \Z, \muALC^{\cap} \) etc.
  Of course, the notion of \emph{abstract DLs} can be formalised by means of abstract model theory, see e.g.~\cite[Sec.~1.2]{PiroPhD12}.
  }

\paragraph*{A bit of graph theory.}\label{subsec:graph-theory}
  We revisit the classical notions of substructures, paths and connectivity.
  Let \( \interI \) be an interpretation. 
  The \newnotion{restriction} of \( \interI \) to a set \( S \subseteq \DeltaI \), is the structure \(\restr{\interI}{S} \) defined by:
  \begin{equation*}
    \label{eq:restriction}
    \DeltaInter{\restr{\interI}{S}} = S, \;
    \roler^{\restr{\interI}{S}} = \roler^{\interI} \cap (S \times S), \;
    \conceptA^{\restr{\interI}{S}} = \conceptA^{\interI} \cap S, \;
    \indva^{\restr{\interI}{S}} = \indva^{\interI} 
    \; \text{if} \; \indva^{\interI} \in S \; \text{otherwise} \; \indva^{\restr{\interI}{S}} \; \text{is undefined,} 
  \end{equation*}
  for all  \(\conceptA \in \Clang \), \(\roler \in \Rlang \) and \(\indva \in \Ilang \).
  A \newnotion{substructure} of \( \interI \) is any of its restrictions~\(\restr{\interI}{S} \) for any \(S \subseteq \DeltaI \).

  The notion of paths is introduced next. 
  An \newnotion{undirected path} (resp. a \newnotion{directed path}) of length \( k{-}1 \) in \(\interI \) is a word~\(\pathrho = \pathrho_1\pathrho_2\ldots\pathrho_k \in (\DeltaI)^{+} \) such that for any index \( i < k \) we have \((\pathrho_i, \pathrho_{i+1}) \in \roler^{\interI} \cup (\roler^{\interI})^{-1} \) for some role name~\( \roler \in \Rlang \) (or just \((\pathrho_i, \pathrho_{i+1}) \in \roler^{\interI} \) in the directed case).
  An element \( \domeleme \in \DeltaI \) is \newnotion{reachable} from \( \domelemd \in \DeltaI \) via an  (un)directed path if there exists an (un)directed path  \( \pathrho = \pathrho_1\pathrho_2\ldots\pathrho_k \) in \( \interI \) with \( \pathrho_1 = \domelemd \) and \( \pathrho_k = \domeleme \).
  We say that \( \interI \) is \newnotion{connected} if any of its domain elements are reachable from any other via an undirected path.
  A structure \( \interJ \) is a \newnotion{connected component} of \( \interI \) if it is a maximal connected substructure of~\( \interI \).
  For any number $k \geq 0$ we define the \emph{$k$-neighbourhood} of $\domelemd$ in $\interI$, denoted with~$\Nbhood{\interI}{k}{\domelemd}$, as the restriction of $\interI$ to elements reachable from $\domelemd$ in $\interI$ by \emph{undirected} paths of length $\leq k$.

  Given a set \(\datadomD \), we say that a structure \(\interI \) is a \(\datadomD \)-\newnotion{forward-forest}, if \(\DeltaI \) is a prefix-closed subset of~\(\datadomD^+ \) and for all \(\roler \in \Rlang \), if \((\domelemd, \domeleme) \in \roler^{\interI} \) then either \(\domelemd,\domeleme \in \datadomD \) or \(\domeleme = \domelemd \cdot \domelemc \) for some \(\domelemc \in \datadomD \).
  The elements of~\(\DeltaI \cap \datadomD \) are called the \newnotion{roots} of \(\interI \). 
  We call \(\interI \) a \(\datadomD \)-\newnotion{forward-tree} if it is a connected \(\datadomD \)-forward-forest with a unique root. 
  We omit the set~\(\datadomD \) and the adjective ``forward'' in the naming whenever it is known from the context or unimportant.
  An interpretation is \emph{forward-tree-shaped} if it is a \(\datadomD \)-forward-tree for some $\datadomD$.

  When working with forests it is convenient to employ the tailored terminology, borrowed from graph theory.
  Given a  \(\datadomD \)-forward-forest \( \interI \) we define an ordering  \( (\DeltaI, \isancestoreq) \) on it in such a way that \( \domelemd \isancestoreq \domeleme \) holds iff \( \domelemd \) is a prefix of \( \domeleme \) and use the following naming scheme:
  \begin{itemize}\itemsep0em
    \item If \( \domelemd \isancestor \domeleme \) holds then \( \domelemd \) is an \newnotion{ancestor} of \( \domeleme \) or, alternatively, \( \domeleme \) is a descendant of \( \domelemd \).
    \item If \( \domelemd_1 \isancestor \domelemd_2 \) but there is no \( \domeleme \) such that \( \domelemd_1 \isancestor \domeleme \isancestor \domelemd_2 \) we call \( \domelemd_1 \) a \emph{parent} of \( \domelemd_2 \) or, alternatively, that \( \domelemd_2 \) is a \emph{child} of \( \domelemd_1 \). Note that it implies that there exists a value \( \domelemc \in \datadomD \) such that \( \domelemd_2 = \domelemd_1 \domelemc \).
    %
    %
    %
    \item The \( \isancestor \)-maximal elements are called \newnotion{leaves}.
    \item Given \( \domelemd \in \DeltaI \) we denote the set of its children and its ancestors, respectively, with \( \childrenof{\interI}{\domelemd} \) and \( \ancestorsof{\interI}{\domelemd} \). We also define the \newnotion{subtree} rooted at \( \domelemd \), denoted: \( \subtreeofrootedat{\interI}{\domelemd} \), \ie{} the restriction of \( \interI \) to the set \( \set{\domelemd} \cup \ancestorsof{\interI}{\domelemd} \).
  \end{itemize}

  To conclude the section, we lift the notion of ``being a forest'' to models of knowledge bases.
  Take a set of individual names \(\names \subseteq \Ilang \). We say that a forward forest \(\interI \) is \newnotion{\( \names \)-rooted} whenever:
  \begin{itemize}\itemsep0em
      \item for all names \(\indva \in \names \) we have that \(\indva^{\interI} \) is defined and it is a root of \(\interI \) and
      \item for each root \(\domelemd \in \DeltaI \) there is a name \(\indva \in \names \) satisfying \(\domelemd = \indva^{\interI} \).
  \end{itemize}
  A \newnotion{forward forest model} of a knowledge base \(\kbK = (\aboxA, \tboxT) \) is an \(\indA \)-rooted forest satisfying \(\kbK \). 
  Abstract DLs $\abstrDL$ for which it is true that every satisfiable $\abstrDL$-KB $\kbK$ has a forward forest model, are said to possess the \emph{forward-forest-model property} (FFMP).
  A prominent example of such a logic is $\ALCcap$.

\paragraph*{Morphisms.}\label{subsec:prelim-morphisms}
  Let \(\interI, \interJ \) be structures and let \(\names \subseteq \Ilang \). 
  An \(\names \)-\newnotion{homomorphism} \(\homof:\interI\to\interJ \) is a function that:
  \begin{itemize}\itemsep0em
      \item maps \(\DeltaI \) to \(\DeltaJ \),
      \item preserves individual names from \(\names \), \ie{} for all \(\indva \in \names \) if \(\indva^{\interI} \) 
      is defined then \(\indva^{\interJ} = \homof(\indva^{\interI}) \),
      \item preserves atomic concepts, \ie{} \(\domelemd \in \conceptA^{\interI} \) implies \(\homof(\domelemd) \in \conceptA^{\interJ} \) for all \(\conceptA \in \Clang \),
      \item and preserves atomic roles, \ie{} \((\domelemd, \domeleme) \in \roler^{\interI} \) implies~\(\left(\homof(\domelemd), \homof(\domeleme)\right) \in \roler^{\interJ} \) for all \(\roler \in \Rlang \). 
  \end{itemize}


\paragraph*{Queries.}\label{subsec:prelim-queries}
  Queries employ \emph{variables} from a countably infinite set \(\Vlang \). 
  A \emph{conjunctive query} (CQ) is a conjunction of \emph{atoms} of the form \(\roler(\varx, \vary) \) or \(\conceptA(\varz) \), where \(\roler \) is a role name, \(\conceptA \) is a concept name and \(\varx, \vary, \varz \) are variables. 
  More expressive query languages are also considered: a \emph{union of conjunctive queries} (UCQ) is a disjunction of CQs and a \emph{positive existential query} (PEQ) is a positive boolean combination of CQs.\footnote{PEQs are generated with the following grammar: $\queryq ::= \conceptA(\varx) \mid \roler(\varx, \vary) \mid \queryq \land \queryq \mid \queryq \lor \queryq$.} 
  Note that any PEQ can be converted to a UCQ of (possibly) exponential size by turning it into disjunctive normal form.

  Let \(\queryq \) be a PEQ and let \(\interI \) be a structure.
  The set of variables appearing in \(\queryq \) is denoted with \(\queryVar{\queryq} \) and the number of atoms of \(\queryq \) (\ie{} the size of \( \queryq \)) is denoted with \(|\queryq| \). 
  The fact that \(\roler(\varx,\vary) \) appears in \(\queryq \) is indicated with~\(\roler(\varx,\vary) \in \queryq \).
  Whenever some subset \( V \subseteq \queryVar{\queryq} \) is given, let \( \restr{\queryq}{V} \) denote the sub-query of \( \queryq \) where all the atoms containing any variable outside \( V \) are removed.
  
  Let \(\matchpi:\queryVar{\queryq}\to\DeltaI \) be a \emph{variable assignment}. 
  We write \(\interI \modelsmatch{\matchpi} \roler(\varx,\vary) \) if~\((\matchpi(\varx),\matchpi(\vary))\in \roler^\interI \) and \(\interI \modelsmatch{\matchpi} \conceptA(\varz) \) if~\(\matchpi(\varz) \in \conceptA^\interI \). 
  Similarly, we write~\(\interI \modelsmatch{\matchpi} \queryq_1 \land \queryq_2 \) (resp. \(\interI \modelsmatch{\matchpi} \queryq_1 \lor \queryq_2 \)) iff \(\interI \modelsmatch{\matchpi} \queryq_1 \) and (resp. or) \(\interI \modelsmatch{\matchpi} \queryq_2 \), for queries \(\queryq_1, \queryq_2 \). 
  We say that~\(\matchpi \) is a \emph{match} for \(\interI \) and \(\queryq \) if \(\interI \modelsmatch{\matchpi} \queryq \) holds and that \(\interI \) \emph{satisfies} \(\queryq \) (denoted with: \(\interI \models \queryq \)) whenever \(\interI \modelsmatch{\matchpi} \queryq \) for some match \(\matchpi \). 
  The definitions are lifted to knowledge bases: \(\queryq \) is \emph{(finitely) entailed} by a~knowledge base \(\kbK \) (written:~\(\kbK \modelsoptfin \queryq \)) if every (finite) model \(\interI \) of~\( \kbK \) satisfies~\(\queryq \). 
  We stress that the entailment relations \(\models \) and~\(\modelsfin \) may not coincide.
  When \(\interI \models \kbK \) but \(\interI \not\models \queryq \), we call \(\interI \) a \emph{countermodel} for \(\kbK \) and \(\queryq \). 
  Note that \(\queryq \) is (finitely) entailed by \(\kbK \) if there is no (finite) countermodel for \(\kbK \) and~\(\queryq \).

  Observe that a \emph{conjunctive query} \(\queryq \) can be seen as a structure \(\interIqueryq = (\queryVar{\queryq}, \cdotIqueryq) \), having the interpretation of roles and concepts fixed as \(\conceptA^{\interIqueryq} = \set{ \varx \mid \conceptA(\varx) \in \queryq} \) and \(\roler^{\interIqueryq} = \set{ (\varx, \vary) \mid \roler(\varx, \vary) \in \queryq} \) for all~\(\conceptA \in \Clang \) and \(\roler \in \Rlang \) and with \(\indva^{\interIqueryq} \) undefined for all \(\indva \in \Ilang \). 
  Hence, any match \(\matchpi \) for \(\interI \) and CQ \(\queryq \) can be seen as an \(\Ilang \)-homomorphism from \(\interIqueryq \) to \(\interI \). 
  We say that a CQ \(\queryq \) is forward-tree-shaped whenever \(\interIqueryq \) is forward-tree-shaped.

\paragraph*{Decision problems.}\label{subsec:prelim-decision-problems}
  For a given description logic \( \abstrDL \) we consider the classical decision problems, namely the (finite) \emph{satisfiability problem} and the (finite) CQ/UCQ/PEQ \emph{entailment problem}. 
  The former asks if an input knowledge base has a (finite) model, while in the latter asks if an input CQ/UCQ/PEQ is (finitely) entailed by an input knowledge base. 
  Here we mention a few results on \(\ALC \) and sister logics. 
  It is well-known that \(\ALC \) has the \emph{finite model property}~\cite[Thm.~3.10]{Gradel99}, \ie{} the satisfiability and the finite satisfiability problems coincide.
  Moreover, \(\ALC \) is \emph{finitely controllable}~\cite[Thm.~1.2]{BaranyGO13} that is, any UCQ is entailed by an \(\ALC \) knowledge base iff it is finitely entailed.
  These two results rely on the fact that \(\ALC \) can be encoded~\cite[Ch.~2.6.1]{dlbook} in the so-called \emph{guarded fragment of first-order logic} \(\GF \)~\cite{AndrekaNB98}.
  Regarding the complexity results, the satisfiability problem~\cite[Thm.~6]{GiacomoL96} and the CQ-entailment problem~\cite[Thm.~1]{LutzDL08} for $\ALC$ (and even $\ALCHQ$) are \(\ExpTime \)-complete, while the PEQ-entailment problem for $\ALC$ was recently shown to be \(\TwoExpTime \)-hard~\cite[Thm.~1]{OrtizS14}. 
  The \( \TwoExpTime \) upper bound can be obtained even for very expressive extensions of \(\ALC \) and regular queries extending PEQs~\cite[Thm.~5.23]{CalvaneseEO14}.
  The UCQ entailment problem for $\ALCH$ is known to be $\ExpTime$-complete~\cite[Thm.~6.5.1]{OrtizPhD10}, while the exact complexity of UCQ-querying for many logics, including $\ALCQ$, is still unknown.
  The absence of such results is even more intriguing in the light of the existing \(\TwoExpTime \)-hardness proofs of~CQ entailment for \(\ALCO \)~\cite[Thm.~9]{NgoOS16}, \(\ALCI \)~\cite[Thm.~2]{Lutz07} and~$\ALCSelf$~\cite[Thm.~8.2]{BednarczykR21}, \ie{} the extensions of \(\ALC \) with \emph{nominals}, \emph{inverses of roles} or \emph{self-loops}. 


\section{Query entailment in locally-forward description logics}\label{sec:query-entailment-for-forward-logics}

In this section, we provide a worst-case optimal algorithm for solving (U)CQ entailment problem for the class of locally-forward abstract DLs.
We first define what locally-forward abstract DLs actually are.

\begin{definition}[locally-forward-forest-like]\label{def:locally-looks-like-forest}
For an $n \in \N$ and an $\names \subseteq \Ilang$ we say that an interpretation $\interI$ is \emph{$(n,\names)$-locally-forward-forest-like} (short: is \emph{$(n,\names)$-lff-like}) iff for any $\domelemd \in \DeltaI$ the $n$-neighbourhood $\Nbhood{\interI}{n}{\domelemd}$ is either forward-tree-shaped or is an $\names'$-rooted forward forest with $\names' = \{ \indva \in \names \mid \indva^{\Nbhood{\interI}{n}{\domelemd}} \; \text{is defined} \}$.
\end{definition}

Locally-forward-forest-like structures are next used as ``coverings'' of (finite) interpretations.
The property below is analogous to the quasi-forest homomorphism-cover property from~\cite[p.~8]{BourhisKR14}.

\begin{definition}[coverable by lffs]\label{def:covered-by-locally-forward-forest-like-structures}
Let $\abstrDL$ be an abstract DL and $\kbK$ be an $\abstrDL$-KB.
We say that $\kbK$ is (finitely) \emph{coverable by locally-forward-forest-like structures} (short: lff-coverable) iff for any (finite) model $\interI \models \kbK$ and every $n \in \N$ there is a (finite) \emph{$(n,\indK)$-lff-like} model $\interJ \models \kbK$ that \emph{covers} $\interI$, \ie any $n$-neighbourhood of $\interJ$ can be $\indK$-homomorphically-mapped to~$\interI$. 
\end{definition}

Finally we employ coverings and lff-like interpretations to define locally-forward DLs.
\begin{definition}\label{def:locally-forward-abstract-dl}
An abstract DL $\abstrDL$ is said to be (finitary) \emph{locally-forward} iff all (finitely) satisfiable $\abstrDL$-KBs are (finitely) lff-coverable.
\end{definition}

From the fact that the set of models for any PEQ is closed under homomorphism it follows that:
\begin{fact}\label{fact:only-lff-like-countermodels-matters}
For any (finitary) locally-forward abstract DL $\abstrDL$, any (finitely) satisfiable $\abstrDL$-KB $\kbK$ and any UCQ $\queryq = \textstyle \bigvee_{i=1}^{m} \queryq_i$, we have that if $\kbK \not\models \queryq$ then there is a (finite) $(|\queryq|,\indK)$-lff-like countermodel for $\kbK$ and $\queryq$.
\end{fact}
\begin{proof}
Let $\interI$ be a countermodel for $\kbK$ and $\queryq$. 
By the fact that $\kbK$ is lff-coverable we infer the existence of $(|\queryq|,\indK)$-lff-like model for $\kbK$ and $\queryq$ that covers $\interI$.
We claim that $\interJ$ is the desired lff-like countermodel $\interJ$ for $\kbK$ and $\queryq$. 
Indeed, if we would have $\interJ \models \queryq$ then $\interJ \modelsmatch{\matchpi} \queryq_i$ (for some index $1 \leq i \leq m$ and a match $\matchpi$).
Then the connected components of $\restr{\interJ}{\{ \matchpi(\varx) \mid \varx \in \queryVar{\queryq_i} \}}$ are of size $\leq |\queryq|$ and hence, can be homomorphically mapped to $\interI$ by the assumption.
This implies $\interI \models \queryq_i$, and therefore $\interI \models \queryq$, contradicting the countermodelhood of $\interI$.
\end{proof}

One can easily provide multiple examples of locally-forward abstract DLs.
It is easy to check that any logic~$\abstrDL$ extending $\ALCcap$ and having the forward-forest-countermodel property is immediately locally-forward, \ie for any~$\kbK$, a guaranteed forward-forest countermodel $\interI$ of $\kbK$ and a (U)CQ of size $n$ is $(n, \indK)$-lff-like.
By inspecting the proof of~\cite[Lemma 3.2.13]{OrtizPhD10} one can see that any abstract DL $\abstrDL$ contained in $\ALCHbregQ$ has such a property and thus, is locally-forward. A more direct proof will be provided in the full version of this paper.
An example of finitary locally-forward abstract DL is $\ALCSCC$~\cite[Lemmas 14--20]{BaaderBR20}.

\subsection{An informal explanation of the Lutz's spoiler technique}\label{subsec:informal-spoilers}

We start by giving a rather informal explanation of Lutz's spoiler technique, dedicated to the readers that are not familiar with the original work of Lutz on querying $\ALCHQ$~\cite[Sec. 3]{LutzDL08}.\footnote{In his paper, Lutz works with $\SHQ$, an extension of $\ALCHQ$ with transitive roles, but he doesn't allow for transitive roles in queries. This is crucial since their presence makes CQ entailment problem exponentially-harder~\cite[Thm. 1]{EiterLOS09}. Hence, from the query point of view, Lutz's work is rather about querying $\ALCHQ$.}
Most of the forthcoming notions are very similar to those from~\cite{LutzDL08} and actually we aimed at reusing as much material from~\cite{LutzDL08} as possible. 
However, many of our statements require separate proofs in order to make them logic-independent and to adjust the proof to locally-forest-like structures.

Recall that our goal is to decide, given an finitely lff-coverable $\abstrDL$-KB $\kbK$ and a conjunctive query $\queryq$, whether $\kbK \modelsoptfin \queryq$ holds, which boils down to checking if there is a (finite or arbitrary, depending on the problem) countermodel for $\kbK$ and $\queryq$.
Due to~\cref{fact:only-lff-like-countermodels-matters} we can restrict our attention to $(n,\indK)$-lff-like interpretations.
An important observation is that a match $\matchpi$ of $\queryq$ over an $(|\queryq|, \indK)$-lff-like $\interI$ induces a very specific partition of $\queryVarq$, namely $\matchpi$ divides the variables of $\queryq$ into three disjoint categories: (i) the variables mapped to the $\names$-named elements of $\interI$, (ii) the variables forming a forward subtree ``dangling'' from one of the $\names$-named elements of $\interI$ and (iii) the variables forming forward-trees that lie ``far'' from $\names$-named elements.
The notion of a \emph{splitting} abstractly describes such a partition, independently of the choice of $\matchpi$ and $\interI$.
The existence of a splitting \emph{compatible} with a $(|\queryq|, \indK)$-lff-like $\interI$ implies that $\interI \models \queryq$ holds and vice-versa.
Hence, to show $\kbK \not\models \queryq$, it suffices to find a $(|\queryq|, \indK)$-lff-like model $\interI^{\spoil}$ of $\kbK$ such that no splitting is compatible with it, or, in other words, that $\interI^{\spoil}$ \emph{spoils} all the splittings.

Next, for a splitting $\splittingofq{}$ of $\queryq$ we design an $\abstrDL$-KB $\spoilerKBof{\splittingofq{}}$, called a \emph{spoiler} for $\splittingofq{}$ with the intended meaning that every $(|\queryq|, \indK)$-locally-forward-forest-like model of $\kbK \cup \spoilerKBof{\splittingofq{}}$ \emph{spoils} its compatibility with~$\splittingofq{}$.
The construction of spoilers employs, among other ingredients, the well-known \emph{rolling-up technique}~\cite[Sec.~4]{HorrocksT00} that is used to detect forward-tree-shaped query matches from points (ii)--(iii) above (the name of the technique comes from the fact that we traverse an input forward-tree in a bottom-up manner and gradually ``rolling-up'' its forward subtrees into predicates, until the root is reached). 
This is the only reason why we require that $\abstrDL$ polynomially encodes $\ALCcap$ concepts.
Having the splittings defined, we observe that (finite) $(|\queryq|, \indK)$-lff-like models of $\kbK \cup \textstyle \bigcup_{\splittingofq{}}\spoilerKBof{\splittingofq{}}$ are also (finite) countermodels for $\kbK$ and $\queryq$.

This yields decidability, but with a suboptimal complexity when the (finite) satisfiability problem for $\abstrDL$  is $\ExpTime$-complete.
To get the optimal (exponential) upper bound in such case, we parallelise the construction of $\bigcup_{\splittingofq{}}\spoilerKBof{\splittingofq{}}$.
This means, intuitively, that the KB $\bigcup_{\splittingofq{}}\spoilerKBof{\splittingofq{}}$ is divided into exponentially many chunks called \emph{super-spoilers} $\superspoilerKBof{\queryq}$ with the meaning that $\kbK \not\modelsoptfin \queryq$ iff $\kbK \cup \superspoilerKBof{\queryq}$ has a (finite) $(|\queryq|, \indK)$-lff-like model for some super-spoiler $\superspoilerKBof{\queryq}$.
We then show that each super-spoiler is of polynomial size and the set of super-spoiler can be enumerated in exponential time. 
This gives us a Turing reduction from the (finite) query entailment problem to exponentially many (finite) satisfiability checks of polynomial-size $\abstrDL$-KBs, which yields an optimal complexity.

\subsection{Step I: Rolling-up forward-tree-shaped queries}\label{subsec:spoiler-technique-rolling-up}

We next recall the well-known rolling-up technique~\cite[p. 5]{LutzDL08} of transforming forward-tree-shaped queries into \( \ALCcap \)-concepts.
Our goal is to construct, for every \( \varx \in \queryVarq \), a concept \( \subtreeConcept{\queryq}{\varx} \) stating that \( \domelemd \in (\subtreeConcept{\queryq}{\varx})^{\interI} \) holds whenever the subtree of \( \interIqueryq \) rooted at the variable \( \varx \) can be mapped below \( \domelemd \) in \( \interI \) (made more formal in~\cref{lemma:rolling-up-trees-induction}). 
A formal, inductive definition is given next. 
The main idea behind the definition is to traverse the input tree in a bottom-up manner, describing its shape with \( \ALCcap \) concepts, and gradually ``rolling-up'' the input forward-tree into smaller chunks until the root is~reached.

\begin{definition}\label{def:rolling-up-concepts} 
    For a forward-tree-shaped CQ \( \queryq \) and any of its variables \( \varv \in \queryVarq \) we define an $\ALCcap$-concept \newnotion{\( \subtreeConcept{\queryq}{\varv} \)} as:
    \[
    \subtreeConcept{\queryq}{\varv} \deff 
      \bigdland_{\conceptA(\varv) \in \queryq} \conceptA \; \; \dland \; \; \bigdland_{\varu \in \Childrenofvar{\varv} } \; \exists \left( \bigcap_{\roler(\varv, \varu) \in \queryq} \roler \right) \subtreeConcept{\queryq}{\varu},
    \]
    where the empty conjunction equals $\top$.
    We set \newnotion{\( \matchConcept{\queryq} \)} as an abbreviation of \( \subtreeConcept{\queryq}{\varv_r} \) with \( \varv_r \) being the root of~$\interIqueryq$.
\end{definition}

From the presented construction we can easily estimate the size (\ie the number of sub-concepts) of~\( \matchConcept{\queryq} \).
Note that the size of  \( \matchConcept{\queryq} \)  is linear in \( |\queryq| \) since every query atom contributes to exactly one sub-concept of  \( \matchConcept{\queryq} \).
The following lemma is folklore and can be shown by routine induction over $(\queryVarq, \isancestoreq)$.
\begin{lemma}\label{lemma:rolling-up-trees-induction}
For any interpretation \( \interI \), any forward-tree-shaped CQ \( \queryq \) and any of its variables \( \varv \in \queryVarq \), the following equivalence holds:  \( \domelemd \in (\subtreeConcept{\queryq}{\varv})^{\interI} \) iff there exists a homomorphism \( \homoh : \subtreeofrootedat{\interI}{\varv}_{\queryq} \to \interI \) with \( \homoh(\varv) = \domelemd \).
\end{lemma}

By unravelling the definition of \( \matchConcept{\queryq} \) and by applying \cref{lemma:rolling-up-trees-induction} for the root variable of \( \queryq \),  we obtain:
\begin{corollary}\label{corr:rolling-up-trees-full}
For any interpretation \( \interI \) and a forward-tree-shaped conjunctive query \( \queryq \) we have \( (\matchConcept{\queryq})^{\interI} \neq \emptyset \) iff there exists a homomorphism \( \homoh : \interIqueryq \to \interI \). 
\end{corollary}

Unfortunately, the presented method of detecting query matches works only for forward-tree-shaped queries.
To detect matches of arbitrary CQs, we introduce the notions of fork rewritings and splittings.

\subsection{Step II: Fork rewritings}\label{subsec:spoiler-technique-fork-rewritings}

Observe that a conjunctive query can induce several different query matches, depending on how its variables ``glue'' together.
We formalise this concept with the forthcoming notion of fork rewritings~\cite[p. 4]{LutzDL08}. 

\begin{definition}\label{def:fork-stuff}
    Let \( \queryq, \queryq' \) be CQs. 
    We say that \( \queryq' \) is obtained from \( \queryq \) by \newnotion{fork elimination}, and denote this fact with \( \queryq \eliminateforkto \queryq' \), if~\( \queryq' \) can be obtained from \( \queryq \) by selecting two atoms \( \roler(\vary, \varx) \), \( \roles(\varz, \varx) \)  of \( \queryq \) (where \( \roler \) and \( \roles \) are not necessary different) and identifying the variables \( \vary, \varz \). 
    We also say that \( \queryq' \) is a  \newnotion{fork rewriting} of \( \queryq \) if \( \queryq' \) is obtained from \( \queryq \) by applying fork elimination on \( \queryq \) possibly multiple times.
    When the fork elimination process is applied exhaustively on \( \queryq \) we say that the resulting query, denoted with \( \maximalforkrew{\queryq} \), is the \newnotion{maximal fork rewriting} of \( \queryq \). 
\end{definition}

The proof of the following~\cref{lemma:unique-maximal-fork-rew} can be found in Appendix A of the technical report for~\cite{LutzDL08}.
\begin{lemma}[Lemma 1 of~\cite{LutzDL08}]\label{lemma:unique-maximal-fork-rew}
    For any conjunctive query \( \queryq \) there exists its (up to a variable renaming) unique maximal fork rewriting \( \maximalforkrew{\queryq} \).
\end{lemma}

To get a better understanding on how the fork elimination works, consult the example below.

\begin{example}\label{exmpl:fork-elimination-example}
    Consider a conjunctive query \( \queryq = \roler(\varx, \vary) \land \roler(\varx, \varz) \land \roles(\varv, \vary) \land \roler(\varv, \varz) \land \conceptA(\varx) \land \conceptB(\vary) \land \conceptC(\varz) \land \conceptD(\varv)\). 
    By~applying fork elimination for variables \( \varx \) and \( \varv \) we obtain the maximal fork rewriting of \( \queryq \), \ie{} the conjunctive query \( \maximalforkrew{\queryq} = \roler(\varx\varv, \vary) \land \roles(\varx\varv, \vary) \land \roler(\varx\varv, \varz) \land \conceptB(\vary) \land \conceptA(\varx\varv) \land \conceptD(\varx\varv) \land \conceptC(\varz) \), with \( \varx\varv \) being a fresh variable.  

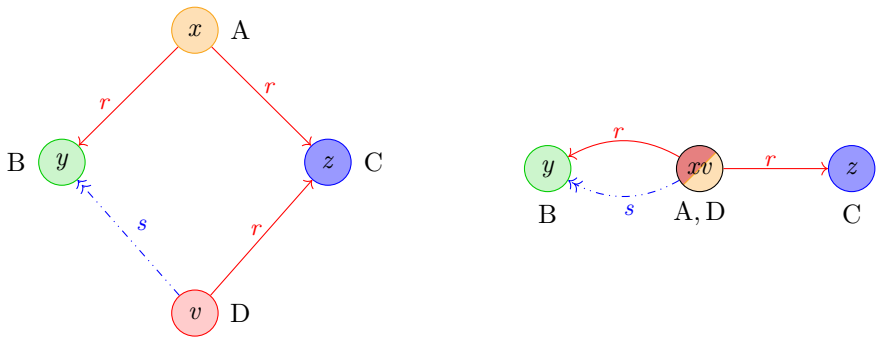
\begin{figure}[H]
          \begin{subfigure}[b]{0.5\textwidth}
            \centering
            \begin{tikzpicture}[transform shape]
                \draw (0, 0) node[ptrond, jaune, label=center:\( \varx \)] (Varx) {\phantom{0}};
                \node[] at (0.6, 0) {\( \conceptA \)};

                \draw (-1.75, -1.75) node[ptrond, vert, label=center: \( \vary \) ] (Vary) {\phantom{0}};
                \node[] at (-2.35, -1.75) {\( \conceptB \)};

                \draw (1.75, -1.75) node[ptrond, bleu, label=center: \( \varz \) ] (Varz) {\phantom{0}};
                \node[] at (2.35, -1.75) {\( \conceptC \)};

                \draw (0, -3.75) node[ptrond, rouge, label=center:\( \varv \)] (Varv) {\phantom{0}};
                \node[] at (0.6, -3.75) {\( \conceptD \)};

                \path[->] (Varx) edge [red] node[yshift=-3, xshift=-9] {\small\(\roler \)} (Vary);
                \path[->] (Varx) edge [red] node[yshift=3, xshift=3] {\small\(\roler \)} (Varz);
                \path[->>,dashdotdotted] (Varv) edge [blue] node[yshift=5, xshift=5] {\small\(\roles \)} (Vary);
                \path[->] (Varv) edge [red] node[yshift=3, xshift=-2] {\small\(\roler \)} (Varz);

            \end{tikzpicture}
          \end{subfigure}
          \hfill
          \begin{subfigure}[b]{0.5\textwidth}
            \begin{tikzpicture}[transform shape]
                \draw (0, 0) node[ptrond, vert, label=center:\( \vary \)] (Vary) {\phantom{0}};
                \node[] at (0, -0.6) {\( \conceptB \)};

                \draw (2, 0) node[ptrond, diagonal fill={jaune}{rouge!50}, label=center: \( \varx\varv \) ] (Varxv) {\phantom{0}};
                \node[] at (2, -0.6) {\( \conceptA, \conceptD \)};

                \draw (4, 0) node[ptrond, bleu, label=center: \( \varz \) ] (Varz) {\phantom{0}};
                \node[] at (4, -0.6) {\( \conceptC \)};

                \path[->] (Varxv) edge [red, bend right=30] node[yshift=3, xshift=-2] {\small\(\roler \)} (Vary);
                \path[->>,dashdotdotted] (Varxv) edge [bend left=30, blue] node[yshift=-5, xshift=2] {\small\(\roles \)} (Vary);
                \path[->] (Varxv) edge [red] node[yshift=3, xshift=-2] {\small\(\roler \)} (Varz);

                \node[] at (0.8, -2) {\phantom{0}};

            \end{tikzpicture}
          \end{subfigure}
          \caption{An example conjunctive query (LHS) and its maximal fork rewriting (RHS).}
\end{figure}
\end{example}

A rather immediate application of Definition~\ref{def:fork-stuff} is that an entailment of a fork rewriting of a query implies the entailment of the input query itself.
The proof goes via an induction over the number of fork eliminations.
\begin{lemma}\label{lemma:entailment-of-fork-rewriting-implies-entailment-of-a-query}
    Let \( \queryq, \queryq' \) be conjunctive queries, such that \( \queryq' \) is obtained from \( \queryq \) by fork rewriting, and let \( \interI \) be a structure. 
    Then \( \interI \models \queryq' \) implies \( \interI \models \queryq \).
\end{lemma}
\begin{proof}
    Assume \( \interI \models \queryq' \).
    Since \( \queryq' \) is a fork rewriting of \( \queryq \), there exists a derivation \( \queryq_n {=} \queryq \eliminateforkto \queryq_{n-1} \eliminateforkto \ldots \eliminateforkto \queryq_0 {=} \queryq' \).
    Reasoning inductively, it suffices to show that for all indices \( 0 \leq i < n \) we have that  \( \interI \models \queryq_i \) implies  \( \interI \models \queryq_{i+1} \). 
    Then we conclude the lemma by taking \( i \deff n-1 \). 
    Assume \( \interI \models \queryq_i \), \ie{}  that there is a homomorphism \( \homoh_i: \interIquery{\queryq_i} \to \interI \).
    Since \( \queryq_{i+1} \eliminateforkto \queryq_i \) holds, we can find the variables \( \varx, \vary, \varz \) such that (i) \( \queryVar{\queryq_i} \setminus \set{\varx, \vary, \varz} = \queryVar{\queryq_{i+1}} \setminus \set{\varx, \vary, \varz} \) and (ii)~\( \queryq_i \) was obtained from \( \queryq_{i+1} \) by replacing each occurrence of \( \varx \) or \( \vary \) in any atoms with \( \varz \).
    Hence, let~\( \homof : \interIquery{\queryq_{i+1}}  \to \interIquery{\queryq_i} \) be a function satisfying \( \homof(\varx) = \homof(\vary) = \varz \) and \( \homof(\varv) = \varv \) for all other variables.
    From (i) and (ii) we immediately infer that \( \homof \) is a homomorphism.
    Thus \( (\homoh_i \circ \homof) : \interIquery{\queryq_{i+1}} \to \interI \) is a homomorphism, establishing \( \interI \models \queryq_{i+1} \).
\end{proof}

\subsection{Step III: Splittings}\label{subsec:spoiler-technique-splittings}
The next notion of splittings~\cite[p. 4]{LutzDL08} provides an abstract way to describe how a conjunctive query $\queryq$ matches a $(|\queryq|, \names)$-locally-forward-forest-like interpretation, while referring neither to a concrete interpretation nor to a concrete match.
Intuitively, the role of splittings is to partition the variables \( \varv \) of some fork rewriting $\queryq$ of the input query, depending on the three possible scenarios:
\begin{itemize}\itemsep0em
    \item either \( \varv \) is mapped to one of the $\names$-named elements,
    \item or \( \varv \), together with some other variables, constitute a subtree dangling from one of the $\names$-named elements,
    \item or \( \varv \) is mapped somewhere ``far'' inside the structure, \ie it is disconnected from the $\names$-named~elements.
\end{itemize}
These intuitions are formalised with a slight modification of the definitions for \( \ALCHQ \) from~\cite[p.~4]{LutzDL08}.

\begin{definition}\label{def:splitting-notion}
    Let \( \names \subseteq \Ilang \) and let \( \queryq \) be a conjunctive query. 
    An \emph{\( \names \)-splitting} \( \splittingofq{\names} \) of \( \queryq \) is a tuple
    \[ \splittingofq{\names} = \left( \splittingroots, \splittingname, \splittingithsubtree{1}, \splittingithsubtree{2}, \ldots, \splittingithsubtree{n}, \splittingrootof, \splittingtrees \right), \]
    where the sets \( \splittingroots, \splittingithsubtree{1}, \ldots, \splittingithsubtree{n}, \splittingtrees \) induce a partition of \( \queryVarq \), \( \splittingname : \splittingroots \to \names \) is a function naming the roots and \( \splittingrootof : \set{1,2,\ldots,n} \to \splittingroots \) assigns to each \( \splittingithsubtree{i} \) an element from \( \splittingroots \).
    Moreover, to be an $\names$-splitting, \( \splittingofq{\names} \) has to satisfy all the conditions below:
    \begin{enumerate}[(a)]\itemsep0em
    \item the query \( \restr{\queryq}{\splittingtrees} \) is a conjunction of variable-disjoint forward-tree-shaped queries,\label{item:a:def:splitting-notion}
    \item the queries \( \restr{\queryq}{\splittingithsubtree{i}} \) are forward-tree-shaped for all indices \( i \in \set{1,2,\ldots,n} \),\label{item:b:def:splitting-notion}
    \item for any atom \( \roler(\varx, \vary) \in \queryq \) the variables \( \varx, \vary \) either belong to the same set or there is an index \( i \in \set{1,2,\ldots,n} \) such that \( \splittingrootof(i) = \varx \in \splittingroots \)  and \( \vary \in \splittingithsubtree{i} \) is the root of \( \restr{\queryq}{\splittingithsubtree{i}} \),\label{item:c:def:splitting-notion}
    \item For any index \( i \in \set{1,2,\ldots,n} \) there is an atom \( \roler(\splittingrootof(i), \varx_{i})  \in \queryq \) with \( \varx_{i} \) being the root of \( \restr{\queryq}{\splittingithsubtree{i}} \).\label{item:d:def:splitting-notion}
    \end{enumerate}
\end{definition}

It helps to think that a splitting consists of named roots, corresponding to the ABox part of the model, together with some of their subtrees and of some auxiliary forward-trees lying somewhere detached from the~roots.

\begin{example}\label{exmpl:splittings}
    Consider an \( \{\indva, \indvb, \indvc\} \)-rooted forward forest \( \interI \) and a (non-tree-shaped) CQ \( \queryq \): 
    \begin{multline*}
        \queryq = \left( \conceptA(\varx_0) \land \roler(\varx_0, \varx_1) \land \roler(\varx_1, \varx_0) \land \conceptB(\varx_1) \right) 
        \land \left( \roles(\varx_0, \varx_{00}) \land \roler(\varx_{00}, \varx_{000}) \right)\\
        \land \left( \roler(\varx_0, \varx_{01}) \land \roles(\varx_{01}, \varx_{010}) \land \roler(\varx_{010}, \varx_{0100}) \right)
        \land \left( \conceptA(\varx_{200}) \land \roler(\varx_{200}, \varx_{2001}) \land \conceptB(\varx_{2001}) \right).
    \end{multline*}
    Then a splitting \( \splittingofq{} = \left( \splittingroots, \splittingname, \splittingithsubtree{1}, \splittingithsubtree{2}, \splittingrootof, \splittingtrees \right) \) defined below is compatible with \( \interI \).

\begin{figure}[H]
    \begin{subfigure}[b]{0.5\textwidth}
    \centering
\begin{tikzpicture}[scale=0.6, transform shape]
  \draw (0,-2) node[ptrond, vert, label=center:\small{00}] (V00) {\phantom{0000}};
  \draw (-2,-4) node[ptrond, jaune, label=center:\small{000}] (V000) {\phantom{0000}};
  \draw (0,-4) node[ptrond, rouge, label=center:\small{001}] (V001) {\phantom{0000}};
  \draw (-2, -6) node[ptrond, vert, label=center:\small{0000}] (V0000) {\phantom{0000}};
  \draw (0, -6) node[ptrond, vert, label=center:\small{0010}] (V0010) {\phantom{0000}};

  \path[->] (V00) edge [red] node[xshift=-6] {\(\roler \)} (V000);
  \path[->] (V00) edge [red] node[xshift=6] {\(\roler \)} (V001);
  \path[->>,dashdotdotted] (V000) edge [blue] node[xshift=-4] {\(\roles \)} (V0000);
  \path[->>,dashdotdotted] (V001) edge [blue] node[xshift=4] {\(\roles \)} (V0010);

  \draw (4, 0) node[ptrond, vert, label=center:\small{0}] (V0) {\phantom{0000}};
  \node[] at (3.25, 0) {\(\indva \)};
  \node[] at (3.25, 0.5) {\(\conceptA \)};
  \path[->>,dashdotdotted] (V0) edge [blue] node[xshift=-10] {\(\roles \)} (V00);

  \draw (4, -2) node[ptrond, jaune, label=center:\small{01}] (V01) {\phantom{0000}};
  \draw (4, -4) node[ptrond, vert, label=center:\small{010}] (V010) {\phantom{0000}};
  \draw (3, -6) node[ptrond, jaune, label=center:\small{0100}] (V0100) {\phantom{0000}};
  \draw (5, -6) node[ptrond, rouge, label=center:\small{0101}] (V0101) {\phantom{0000}};

  \path[->] (V0) edge [red] node[xshift=-5] {\(\roler \)} (V01);
  \path[->>,dashdotdotted] (V01) edge [blue] node[xshift=4] {\(\roles \)} (V010);
  \path[->] (V010) edge [red] node[xshift=-4] {\(\roler \)} (V0100);
  \path[->] (V010) edge [red] node[xshift=4] {\(\roler \)} (V0101);

  \draw (9.25, 0) node[ptrond, vert, label=center:\small{2}] (V2) {\phantom{0000}};
  \node[] at (10, 0) {\(\indvc \)};

  \draw (9.25, -2) node[ptrond, jaune, label=center:\small{20}] (V20) {\phantom{0000}};
  \draw (9.25, -4) node[ptrond, vert, label=center:\small{200}] (V200) {\phantom{0000}};
  \node[] at (10.25, -4) {\( \conceptA \)};

  \draw (8.25, -6) node[ptrond, jaune, label=center:\small{2000}] (V2000) {\phantom{0000}};
  \draw (10.25, -6) node[ptrond, rouge, label=center:\small{2001}] (V2001) {\phantom{0000}};
  \node[] at (11.25, -6) {\( \conceptB \)};

  \path[->] (V2) edge [red] node[xshift=-5] {\(\roler \)} (V20);
  \path[->>,dashdotdotted] (V20) edge [blue] node[xshift=4] {\(\roles \)} (V200);
  \path[->] (V200) edge [red] node[xshift=-4] {\(\roler \)} (V2000);
  \path[->] (V200) edge [red] node[xshift=4] {\(\roler \)} (V2001);

  \draw (6.5, 1) node[ptrond, rouge, label=center:\small{1}] (V1) {\phantom{0000}};
  \node[] at (6.5, 1.75) {\(\indvb, \conceptB \)};

  \path[->>,dashdotdotted] (V1) edge [loop below, blue] node[xshift=-15, yshift=9] {\(\roles \)} (V1);
  \path[->] (V0) edge [bend left=30, red] node[yshift=4] {\small\(\roler \)} (V1);
  \path[->] (V1) edge [bend left=30, red] node[yshift=4.5,xshift=2] {\small\(\roler \)} (V0);
  \path[->] (V2) edge [bend left=30, red] node[yshift=4.5,xshift=2] {\small\(\roler \)} (V0);
  \path[->] (V2) edge [bend right=30, red] node[xshift=4, yshift=4] {\small\(\roler \)} (V1);
  \path[->>,dashdotdotted] (V1) edge [bend right=30, blue] node[xshift=4,yshift=4] {\small\(\roles \)} (V2);




\scoped[on background layer] \filldraw [red!50, draw opacity=0.2, fill opacity=0.2, line width=2.5em, line join=round,] (V0.center) -- (V1.center) -- cycle;
\scoped[on background layer] \filldraw [blue!10, line width=2.5em, line join=round,] (V01.center) -- (V010.center) -- (V0100.center) -- cycle;
\scoped[on background layer] \filldraw [blue!50, draw opacity=0.2, fill opacity=0.2, line width=2.5em, line join=round,] (V00.center) -- (V000.center) -- cycle;
\scoped[on background layer] \filldraw [green!80, draw opacity=0.2, fill opacity=0.2, line width=2.5em, line join=round,] (V200.center) -- (V2001.center) -- cycle;

\end{tikzpicture}
      \end{subfigure}
      \hfill
      \begin{subfigure}[b]{0.5\textwidth}
        \begin{tikzpicture}[transform shape]
            \node[] at (4,  4) {\( \splittingroots = \set{ \varx_0, \varx_1 } \)};
            \node[] at (4,  3.2) {\( \splittingithsubtree{1} = \set{\varx_{00}, \varx_{000}} \)};
            \node[] at (4, 2.4) {\( \splittingithsubtree{2} = \set{\varx_{01}, \varx_{010}, \varx_{0100}} \)};
            \node[] at (4, 1.6) {\( \splittingtrees = \set{\varx_{200}, \varx_{2001}} \)};
            \node[] at (4, 0.8) { \( \splittingname(\varx_0) = \indva, \splittingname(\varx_1) = \indvb \) };
            \node[] at (4, 0.0) { \( \splittingrootof(1) = \varx_0, \splittingrootof(2) = \varx_0 \) };
            \node[] at (0.0, -0.5) { \phantom{0} };
        \end{tikzpicture}
      \end{subfigure}
    \caption{An example splitting \( \splittingofq{} \) of \( \queryq \), compatible with \( \interI \). }
\end{figure}
\end{example}

We finish the section by showing that splittings indeed fulfil their purposes.
In order to do it, we first introduce an immediate definition of \emph{compatibility} of a splitting with a $(|\queryq|, \names)$-locally-forward-forest-like interpretation.

\begin{definition}\label{def:splittings-compatible}
    Let \( \names \subseteq \Ilang \), \( \queryq \) be a CQ and \( \interI \) be a $(|\queryq|, \names)$-locally-forward-forest-like interpretation.
    We say that an \( \names \)-splitting \( \splittingofq{\names} \) of \( \queryq \) is \newnotion{compatible} with \( \interI \) if it satisfies:
    \begin{enumerate}[(A)]\itemsep0em
        \item for every connected component \( \hat{\queryq} \) of \( \restr{\queryq'}{\splittingtrees} \) there is a domain element \( \domelemd \in \DeltaI \) satisfying \( \domelemd \in (\matchConcept{\hat{\queryq}})^{\interI} \),\label{item:A:def:splittings-compatible}
        \item for all atoms \( \conceptA(\varx) \in \queryq \) with \( \varx \in \splittingroots \) we have \( (\splittingname(\varx))^{\interI} \in \conceptA^{\interI} \),\label{item:B:def:splittings-compatible}
        \item for all atoms \( \roler(\varx, \vary) \in \queryq \) with  \( \varx, \vary \in \splittingroots \) we have \( \left( \splittingname(\varx)^{\interI}, \splittingname(\vary)^{\interI} \right) \in \roler^{\interI} \),\label{item:C:def:splittings-compatible}
        \item for all indices \( i \in \set{1,2,\ldots,n} \) the following property, for \(\varx_i\) being the root of \( \restr{\queryq}{\splittingithsubtree{i}}\), is satisfied:\label{item:D:def:splittings-compatible}
        \[ \splittingname(\splittingrootof(i))^{\interI} \in \left( \exists \left( \bigcap_{\roler(\splittingrootof(i), \varx_i) \in \queryq} \roler \right) \matchConcept{\restr{\queryq}{\splittingithsubtree{i}}} \right)^{\interI}. \]
    \end{enumerate}
\end{definition}

The forthcoming lemmas link together all the notions presented in this section. 
Its proof is similar to~\cite[Lemma~2]{LutzDL08}, but our version is arguably more detailed and uses a different kind of structures than Lutz's.

\begin{lemma}\label{lemma:splittings-linking-notions-together-from-splitting-to-match}
    Let \( \queryq \) be a CQ, \( \names \subseteq \Ilang \) and \( \interI \) be a (finite) $(|\queryq|, \names)$-lff-like interpretation.
    Then \( \interI \modelsoptfin \queryq \) if and only if there is a fork rewriting \( \queryq' \) of \( \queryq \) and an \( \names \)-splitting \( \splittingof{\queryq'}{\names} \) of \( \queryq' \), such that \( \splittingof{\queryq'}{\names} \) is compatible with \( \interI \).
\end{lemma}

We start with the ``if'' direction.

\begin{proof}[Proof ($\Leftarrow$)]
  By \cref{lemma:entailment-of-fork-rewriting-implies-entailment-of-a-query}, it suffices to show~\( \interI \models \queryq' \).
  We construct a function \( \homoh : \queryVar{\queryq'} \to \interI \) as follows:
  \begin{itemize}\itemsep0em
    \item For every root variable \( \varx \in \splittingroots \) we put \( \homoh(\varx) := (\splittingname(\varx))^{\interI}\). 
    \item Fix an index \( 1 \leq i \leq n \). By \cref{item:b:def:splitting-notion} of \cref{def:splitting-notion} we know that \( \restr{\queryq'}{\splittingithsubtree{i}} \) is forward-tree-shaped and let \( \varx_i \) be its root.
    Moreover, by \cref{item:D:def:splittings-compatible} of \cref{def:splittings-compatible} there exists an element \( \domelemd_i \in \DeltaI \) satisfying:
    \[ 
    \left( \splittingname(\splittingrootof(i))^{\interI}, \domelemd_i \right) \in \left( \bigcap_{\roler(\splittingrootof(i), \varx_i) \in \queryq'} \roler^{\interI} \right) \quad \text{and} \quad \domelemd_i \in \left( \matchConcept{\restr{\queryq'}{\splittingithsubtree{i}}} \right)^{\interI}. \label{eq:proof-of-lemma:splittings-linking-notions-together-case-2-defining-homo} \tag{\mbox{\(\spadesuit\)}} 
    \]
    From the forward-tree-shapedness of \( \restr{\queryq'}{\splittingithsubtree{i}} \) and \cref{lemma:rolling-up-trees-induction} we conclude the existence of a homomorphism \( \homoh_i \) from \( \interIquery{\restr{\queryq'}{\splittingithsubtree{i}}} \) to \( \interI \) with \( \homoh_i(\varx_i) = \domelemd_i \).
    Thus we can simply put \( \homoh(\varx) := \homoh_i(\varx) \) for all \( \varx \in \splittingithsubtree{i} \). 
    \item Take any connected component \( \hat{\queryq} \) of \( \restr{\queryq'}{\splittingtrees} \), which by \cref{item:a:def:splitting-notion} of \cref{def:splitting-notion} is forward-tree-shaped.
    From the compatibility of \( \splittingof{\queryq'}{\names} \) with \( \interI \) and \cref{item:A:def:splittings-compatible} of \cref{def:splittings-compatible} we know that there is an element \( \domelemd \in \DeltaI \) satisfying \( \domelemd \in (\matchConcept{\hat{\queryq}})^{\interI} \).
    Invoking \cref{corr:rolling-up-trees-full}, we deduce that there exists a homomorphism \( \homoh_{\hat{\queryq}} : \interIquery{\hat{\queryq}} \to \interI \).
    Finally, we put \( \homoh(\varx) := \homoh_{\hat{\queryq}}(\varx) \) for all \( \varx \in \queryVar{\hat{\queryq}} \).
  \end{itemize}
  Note that the definition of $\homoh$ is correct, \ie that every argument has a value assigned and that each argument has only one value assigned, since (i) the sets \( \splittingroots, \splittingithsubtree{1}, \ldots, \splittingithsubtree{n}, \splittingtrees \) induce a partition of \( \queryVarq \), (ii) that all forward-tree-shaped queries from \( \splittingtrees \) are variable-disjoint and (iii) the employed homomorphism are functions themselves. 
  Hence, it remains to show that \( \homoh \) is also a homomorphism from $\interIquery{\queryq'}$ to $\interI$.
  Proving the preservation of atomic concepts by \( \homoh \) is immediate: for root variables we employ \cref{item:B:def:splittings-compatible} of \cref{def:splittings-compatible}, while for the other variables we rely on the fact that the result of \( \homoh \) is then defined via an another homomorphism.
  For the proof of the preservation of roles by \( \homoh \), we take any~\( (\varx, \vary) \in \roler^{\interIquery{\queryq'}} \), or equivalently \( \roler(\varx,\vary) \in \queryq' \), and we going to show that \( (\homoh(\varx), \homoh(\vary)) \in \roler^{\interI}  \).
  By \cref{item:c:def:splitting-notion} of \cref{def:splitting-notion} we know that there are only four cases to consider, depending on the location of \( \varx\) and \( \vary \):
  \begin{itemize}\itemsep0em
      \item Both \(\varx \) and \( \vary \) belong to \( \splittingroots \).\\ 
      Then \( \left( \homoh(\varx), \homoh(\vary) \right) = \left( \splittingname(\varx)^{\interI}, \splittingname(\vary)^{\interI} \right) \in \roler^{\interI} \) follows from \cref{item:C:def:splittings-compatible} of \cref{def:splittings-compatible}.

      \item There exists an index \( 1 \leq i \leq n \) such that \( \varx, \vary \in \splittingithsubtree{i} \).\\
      Then \( (\varx, \vary) \in \roler^{\interIquery{\restr{\queryq'}{\splittingithsubtree{i}}}} \) holds and we get \( \left( \homoh(\varx), \homoh(\vary) \right) = \left( \homoh_i(\varx), \homoh_i(\vary) \right) \in \roler^{\interI} \) since \( \homoh_i \) is a homomorphism.

      \item Both \(\varx \) and \( \vary \) belong to \( \splittingtrees{} \).\\
      From \( (\varx, \vary) \in \roler^{\interIquery{\queryq'}} \) we know that \( \varx, \vary \) are in the same subtree \( \hat{\queryq} \) of \( \splittingtrees \).
      Thus, \( (\varx, \vary) \in \roler^{\interIquery{\hat{\queryq}}} \) holds and it suffices to apply the fact that \( \homoh_{\hat{\queryq}} \) is homomorphism to get \( \left( \homoh(\varx), \homoh(\vary) \right) = \left( \homoh_{\hat{\queryq}}(\varx), \homoh_{\hat{\queryq}}(\vary) \right) \in \roler^{\interI} \).
   
      \item The variables \( \varx \) and \( \vary \) are in two different sets.\\
      First, from \cref{item:c:def:splitting-notion} of \cref{def:splitting-notion}, we know that there is an \( i \) such that \( \varx \in \splittingroots \) satisfies \( \splittingrootof(i) = \varx \) and \( \vary = \varx_i \in \splittingithsubtree{i} \) is the root of \( \restr{\queryq}{\splittingithsubtree{i}} \).
      Second, by \cref{eq:proof-of-lemma:splittings-linking-notions-together-case-2-defining-homo} we know that~\( (\homoh(\varx), \homoh(\vary)) \) is actually equal to \( \left( \splittingname(\splittingrootof(i))^{\interI}, \domelemd_i \right) \) for some already-fixed \( \domelemd_i \in \DeltaI \). 
      Finally, by applying the first part of \cref{eq:proof-of-lemma:splittings-linking-notions-together-case-2-defining-homo}, we get \( (\homoh(\varx), \homoh(\vary)) = \left( \splittingname(\splittingrootof(i))^{\interI}, \domelemd_i \right)  \) belongs to \( \roler^{\interI}\), as required.
  \end{itemize}
  We have shown that the function \( \homoh \) is indeed a homomorphism.
  Thus \( \interI \models \queryq' \) holds, implying \( \interI \models \queryq \). \qedhere
\end{proof}

Next, we proceed with a more difficult ``only if'' direction of~\cref{lemma:splittings-linking-notions-together-from-splitting-to-match}.

\begin{proof}[Proof ($\Rightarrow$)]
  Let $\matchpi$ be the match witnessing $\interI \modelsmatch{\matchpi} \queryq$.
  We construct the query \( \queryq' \) by exhaustively applying fork elimination on all ``forks'' \( \roler(\varx, \varz), \roles(\vary, \varz)\) (where $\roler,\roles$ are not necessarily different) with \( \matchpi(\varx) = \matchpi(\vary) \).
  Note that then also $\interI \models \queryq'$ holds (a match $\matchpi'$ for $\queryq'$ can be easily constructed from $\matchpi$).
  In what follows we define an \( \names \)-splitting \[ \splittingof{\queryq'}{\names} = \left( \splittingroots, \splittingname, \splittingithsubtree{1}, \splittingithsubtree{2}, \ldots, \splittingithsubtree{n}, \splittingrootof, \splittingtrees \right),\] where the definitions of its components are provided below.
  \begin{itemize}\itemsep0em
      \item The set \( \splittingroots \) is composed of all variables \( \varx \in \queryVar{\queryq'} \) for which \( \matchpi'(\varx) \) is an $\names$-named element of $\interI$. 
      For all such variables $\varx$ we set \( \splittingname(\varx) = \indva \) for any corresponding \( \indva \in \names \). 
      \item The sets \( \splittingithsubtree{i} \), as their name suggests, are defined by taking subtrees connected to the roots.
      To simplify the definition, we say that a variable \( \varx \) is \emph{dangling from a root} if there exists a variable \( \varx_r \in \splittingroots \) and an atom \( \roler(\varx_r, \varx) \) in \( \queryq' \).
      Let \( D \) be the subset-maximal set of variables from \( \left( \queryVar{\queryq'} \setminus \splittingroots \right) \) dangling from roots.
      Take \( n \deff |D|\) and fix an ordering \( \varx_1, \varx_2, \ldots, \varx_n \) on the elements from \( D \).
      For any index \( 1 \leq i \leq n \) we define \( \splittingithsubtree{i} \) as the set composed of \( \varx_i \) and all variables reachable from \( \varx_i \) via a directed path of positive length in the query structure \(  \interIquery{\restr{\queryq'}{\queryVar{\queryq'} \setminus \splittingroots}} \).
      Observe that $\restr{\interI}{\{ \matchpi'(\varv) \mid \varv \in \splittingithsubtree{i} \}}$ is a forward-tree. 
      This follows from the fact that the $|\queryq|$-neighbourhood of $\matchpi'(\varx_i)$ in $\interI$ is either a forward-tree (and hence we are done) or it is an $\names'$-rooted forward-forest (then since $\matchpi'(\varx_i)$ is not $\names$-named it is an inner node of the forest and hence the nodes reachable from it constitute a forward-tree).

      Thus, due to the fact that we eliminated all forks, the underlying query $\restr{\queryq'}{\splittingithsubtree{i}}$ is forward-tree-shaped.

      \item We put \( \splittingrootof(i) \deff \varx_{r}^{i} \), where \( \varx_{r}^{i} \) is the root from which \( \varx_i \in D \) is dangling.\\ 
      Note that \( \varx_{r}^{i} \)  is uniquely determined due to the construction of \( \queryq' \).
      Indeed, ad absurdum assume that there is \( \vary_{r}^{i} \neq \varx_{r}^{i} \) such that~\( \roler(\varx_{r}^{i}, \varx_i) \) and \( \roles(\vary_{r}^{i}, \varx_i) \) holds.
      There are two cases: either \( \matchpi'(\varx_{r}^{i}) = \matchpi'(\vary_{r}^{i}) \) or \( \matchpi'(\varx_{r}^{i}) \neq \matchpi'(\vary_{r}^{i}) \).
      The former case is clearly not possible due to the fact that such ``forks'' were eliminated in \( \queryq' \).
      In the latter case it implies that there are two $\names$-named elements of \( \interI \) pointing at \( \matchpi'(\varx_i) \).
      Recall that \( \interI \) is a $(|\queryq|, \names)$-locally-forward-forest-like and hence, the local neighbourhood of $\matchpi'(\varx_i)$ is a forward-forest with at least two roots, $\matchpi'(\varx_{r}^{i})$ and $\matchpi'(\vary_{r}^{i})$. 
      Thus the latter case is only possible when \( \matchpi'(\varx_i)  \) is also $\names$-named, but it is not because \( \varx_i \not\in \splittingroots{} \). 
      A contradiction. 
      \item The set \( \splittingtrees \) contains all other variables from \( \queryVar{\queryq'} \).
  \end{itemize}

  Now we show that \( \splittingof{\queryq'}{\names} \) is indeed a splitting. 
  We have already argued that $\splittingname$ and $\splittingrootof$ is are functions and that \cref{item:b:def:splitting-notion} and \cref{item:d:def:splitting-notion} of \cref{def:splitting-notion} hold.
  It remains to prove that the selected sets induce a partition of $\queryVar{\queryq'}$ as well as the satisfaction of \cref{item:a:def:splitting-notion,item:c:def:splitting-notion} of \cref{def:splitting-notion}.

  We start from the former issue.
  First, note that by the above definitions the set \( \splittingtrees \) guarantees that all the set components of \( \splittingof{\queryq'}{\names} \) sums to \( \queryVar{\queryq'} \) and that~\( \splittingtrees \) are disjoint from the other sets.
  Moreover, since the variables from \( \splittingroots \) were excluded while defining~\( \splittingithsubtree{i} \) we conclude that \( \splittingroots \cap \splittingithsubtree{i} = \emptyset \) for any index \( i \).
  Hence, it suffices to take any two indices \( i < j \) and show the disjointness of~\( \splittingithsubtree{i} \) and \( \splittingithsubtree{j} \).
  Assume towards a contradiction that $\splittingithsubtree{i} \cap \splittingithsubtree{j} \neq \emptyset$. 
  Thus there is a variable $\varv$ reachable from both $\varx_i$ and $\varx_j$ (the roots of $\restr{\queryq}{\splittingithsubtree{i}}$ and~$\restr{\queryq}{\splittingithsubtree{j}}$, different by definition) via directed paths in \(  \interIquery{\restr{\queryq'}{\splittingithsubtree{i}}} \) and~\(  \interIquery{\restr{\queryq'}{\splittingithsubtree{j}}}\). 
  We consider the following cases:
  \begin{enumerate}\itemsep0em
    \item $\splittingithsubtree{i} = \{ \varx_i \}$ and $\splittingithsubtree{j} = \{ \varx_j \}$.\\
    This implies that $\varv = \varx_j$ or $\varv = \varx_i$ and contradicts the fact that each of the above sets is a singleton.
    \item $\varv = \varx_i$ (the case of $\varv = \varx_j$ is analogous).\\
    Hence, we infer the existence of a directed path from $\varx_j$ to~$\varx_i$ in \(  \interIquery{\restr{\queryq'}{\splittingithsubtree{j}}} \).
    Moreover, all the elements on this path are anonymous, since they belong to $\splittingithsubtree{j}$.
    Note that $\varx_i$ is also anonymous.
    Thus there is also a directed path (of positive length!) from $\matchpi'(\varx_j)$ to $\matchpi'(\varx_i)$ of length $\leq |\queryq'|$ in~$\interI$.
    But it contradicts the fact that the $|\queryq|$-neighbourhood of $\matchpi'(\varx_i)$ is an $\names'$-forward-forest, with some $\names'$ containing $\splittingname(\splittingrootof(i))$ and~$\splittingname(\splittingrootof(j))$.
    \item $\varx_i \neq \varv \neq \varx_j$.\\
    Thus there are variables $\varu \in \splittingithsubtree{i}$, $\varw \in \splittingithsubtree{j}$ and $\varz \in \splittingithsubtree{i} \cap \splittingithsubtree{j}$ (with $\varz$ possibly equal to $\varv$) such that $\roler(\varu, \varz) \in \queryq'$ and $\roles(\varw, \varz) \in \queryq'$ and $\varz$ is reachable from both $\varx_i$ and $\varx_j$. 
    Since we eliminated all the forks, we know that $\matchpi'(\varw) \neq \matchpi'(\varu)$.
    Thus, by applying the fact that $|\queryq|$-neighbourhood of $\matchpi'(\varx_i)$ is an $\names'$-forward-forest, we get a contradiction because $\matchpi'(\varu), \matchpi'(\varw), \matchpi'(\varz)$ do not form a forward-forest.
  \end{enumerate}

  Hence components of \( \splittingof{\queryq'}{\names} \) indeed induce a partition of $\queryVar{\queryq'}$.

  Next, we proceed with \cref{item:a:def:splitting-notion}.
  Take any connected component $\hat{\queryq}$ of $\restr{\queryq'}{\splittingtrees}$.
  Note that $|\hat{\queryq}| \leq |\queryq|$ and for any variable $\varv \in \queryVar{\hat{\queryq}}$ we have that $\matchpi'(\varv)$ is not $\names$-named.
  Hence, from the $(|\queryq|, \names)$--lff-likeness of $\interI$ we infer the substructure induced by $\matchpi'$ and $\hat{\queryq}$ is a forward-tree, so is~$\hat{\queryq}$ (we eliminated all the forks!).

  Finally, we need to argue that \cref{item:c:def:splitting-notion} of \cref{def:splitting-notion} holds.
  If $\varx \in \splittingroots$ and $\roler(\varx,\vary) \in \queryq'$ then either if $\matchpi'(\vary)$ is $\names$-named then $\vary \in \splittingroots$ (thus $\varx,\vary$ are in the same set) or $\vary$ is dangling from the root so, by the construction, is in some $\splittingithsubtree{i}$.
  By construction, $\vary$ is the root of $\restr{\queryq'}{\splittingithsubtree{i}}$.
  Otherwise $\varx \not\in \splittingroots$ and we consider the following cases:
  \begin{enumerate}\itemsep0em
      \item If $\varx \in \splittingithsubtree{i}$ and $\vary \in \splittingithsubtree{j} \cup \splittingtrees$ then $\vary \in \splittingithsubtree{i}$  violating the disjointness of these sets.
      \item $\vary \in \splittingroots$ or ($\varx \in \splittingtrees$ and $\vary \in \splittingithsubtree{i}$). 
      We get a contradiction with lff-likeness of $\interI$.
  \end{enumerate}
  
  This finishes the proof that \( \splittingof{\queryq'}{\names} \) is an $\names$-splitting of $\queryq'$.
  Next, we will argue that $\splittingof{\queryq'}{\names}$ is compatible with~$\interI$.
  \cref{item:A:def:splittings-compatible} follows from~\cref{corr:rolling-up-trees-full}.
  \cref{item:B:def:splittings-compatible,item:C:def:splittings-compatible} are immediate by the fact that $\interI \modelsmatch{\matchpi'} \queryq'$.
  Finally, for~\cref{item:D:def:splittings-compatible} we take $\varx_i$ (the $i$-th variable dangling from the roots) combine the fact that $\matchpi$ is a homomorphism, thus all the relations mentioned in $\queryq'$ between $\matchpi'(\splittingrootof(i))$ and $\matchpi'(\varx_i)$ are preserved, with~\cref{lemma:rolling-up-trees-induction} to infer that $\matchpi'(\varx_i) \in \matchConcept{\restr{\queryq}{\splittingithsubtree{i}}}^{\interI}$.
  This concludes the proof.
\end{proof}

Following Lutz, we say that a role conjunction  $\roles_1 \cap \ldots \cap \roles_n$ \emph{occurs} in a CQ $\queryq$ if we can find two variables~$\varv, \varv' \in \queryVarq$ such that $\{ \roler \in \Rlang \mid \roler(\varv,\varv') \in \queryq \} = \{ \roles_1, \roles_2, \ldots, \roles_n\}$.
Similarly, we speak about concept/role names occurring in $\queryq$.
Note that the role conjunctions and concept/role names used in~\cref{def:splittings-compatible} occur in $\queryq$. 

We will next link maximal fork rewritings and splittings.
Let $\queryq$ be a CQ and let $\QTree(\maximalforkrew{\queryq})$ denote the set of all forward-forest-shaped queries $\restr{\maximalforkrew{\queryq}}{\Reach(\varv)}$, where $\varv \in \queryVar{\maximalforkrew{\queryq}}$ and $\Reach(\varv)$ denotes the set of all variables reachable from $\varv$ in $\interIquery{\maximalforkrew{\queryq}}$ via a directed path.
Note that the size of $\QTree(\maximalforkrew{\queryq})$ is polynomial in the size of $\maximalforkrew{\queryq}$, thus also in $|\queryq|$.
The following lemma was shown in Appendix A of~\cite{LutzDL08}.

\begin{lemma}\label{lemma:query-trees-are-contained-in-max-fork-rew}
Let $\splittingof{\queryq'}{\names} = \left( \splittingroots, \splittingname, \splittingithsubtree{1}, \ldots, \splittingithsubtree{n}, \splittingrootof, \splittingtrees \right)$ be an $\names$-splitting of $\queryq'$, a fork rewriting of a CQ $\queryq$, let $\queryq_1', \queryq_2', \ldots, \queryq_k'$ be the disconnected components of $\restr{\queryq'}{\splittingtrees}$, and let $\varx_1, \varx_2, \ldots, \varx_n$ be the root variables of the corresponding $\restr{\queryq'}{\splittingithsubtree{i}}$. 
Then: 
\begin{itemize}\itemsep0em
  \item $\queryq_i' \in \QTree(\maximalforkrew{\queryq})$ for all $1 \leq i \leq k$,
  \item for all $1 \leq i \leq n$ we have $\restr{\queryq'}{\splittingithsubtree{i}} \in \QTree(\maximalforkrew{\queryq})$, and
  \item $\bigcap_{\roler \in \{ \roler \mid \roler( \splittingrootof(i), \varx_i ) \in \queryq' \}}.\roler$ occurs in $\maximalforkrew{\queryq}$.
\end{itemize}
\end{lemma}
\begin{proof}
The same as the proof of Lemma 4 in~\cite{LutzDL08} assuming the naming convention from the appendix A.\footnote{Watch out! There is a glitch in Lutz's proof. In the inductive assumption no. 4, there should be $\{ \varv, \varv'\} \neq \{ \varv, \varv'\} \cap S_j \neq \emptyset$ rather than $\{ \varv, \varv'\} \cap S_j \neq \emptyset$. }
\end{proof}

\subsection{Step IV: Spoilers}

Spoilers~\cite[p. 6]{LutzDL08} are \( \ALCcap \)-KBs dedicated for blocking query matches over lff-like structures.

\begin{definition}\label{def:spoiler}
    Let \( \names \subseteq \Ilang \), \( \queryq \) be a CQ and let \( \splittingofq{\names} = (\splittingroots, \splittingname, \splittingithsubtree{1}, \ldots, \splittingithsubtree{n}, \splittingrootof, \splittingtrees) \) be an \( \names \)-splitting \( \splittingofq{\names} \) of \( \queryq \). 
    An $\names$-\emph{spoiler} \( \spoilerKBof{\splittingofq{\names}} \) for \( \splittingofq{\names} \) is an \( \ALCcap \)-KB satisfying at least one of:
    \begin{enumerate}[(A)]\itemsep0em
    \item \( \left( \topconcept \dlsubseteq \neg \matchConcept{\hat{\queryq}}  \right) \in \spoilerKBof{\splittingofq{\names}} \) for some forward-tree-shaped query \( \hat{\queryq} \), a connected component of \( \restr{\queryq}{\splittingtrees} \),\label{item:A:def:spoiler-notion}
    \item \( \left( \neg \conceptA(\splittingname(\varx)) \right) \in \spoilerKBof{\splittingofq{\names}} \) for some atom \( \conceptA(\varx) \in \queryq \) with \( \varx \in \splittingroots \),\label{item:B:def:spoiler-notion}
    \item \( \left( \neg \roler(\splittingname(\varx), \splittingname(\vary)) \right) \in \spoilerKBof{\splittingofq{\names}} \) from some atom \( \roler(\varx, \vary) \in \queryq \) with \( \varx, \vary \in \splittingroots \),\label{item:C:def:spoiler-notion}
    \item \( \left( \neg \exists \left( \bigcap_{\roler(\splittingrootof(i), \varx_i) \in \queryq} \roler \right) \matchConcept{\restr{\queryq}{\splittingithsubtree{i}}} \right) \left( \splittingname(\splittingrootof(i)) \right) \in \spoilerKBof{\splittingofq{\names}} \) for some index \( 1 \leq i \leq n \) (where \( \varx_i \) denotes the root variable of \( \restr{\queryq}{\splittingithsubtree{i}} \)).\label{item:D:def:spoiler-notion}
    \end{enumerate}
\end{definition}

Observe a tight correspondence between Items~\ref{item:A:def:spoiler-notion}--\ref{item:D:def:spoiler-notion} of the above definition and Items~\ref{item:A:def:splittings-compatible}--\ref{item:D:def:splittings-compatible} from \cref{def:splitting-notion}.
We may see these cases as potential ways of ``blocking'' compatibility of a given splitting.

\begin{definition}\label{def:super-spoiler}
    Let \( \names \subseteq \Ilang \) and let \( \queryq \) be a CQ\@. 
    An \( \ALCcap \)-KB \( \superspoilerKBof{\queryq} \) is an \( \names \)-super-spoiler for \( \queryq \) if it is a $\subseteq$-minimal KB such that for all fork rewritings \( \queryq' \) of \( \queryq \) and all \( \names \)-splittings \( \splittingof{\queryq'}{\names} \) of \( \queryq' \) we have that \( \superspoilerKBof{\queryq} \) is~an~$\names$-spoiler~for~\( \splittingof{\queryq'}{\names} \).
\end{definition}

The forthcoming lemma shows that the existence of an $\names$-super-spoiler ``spoils'' the (finite) entailment of an input CQ over (finite) $(|\queryq|,\names)$-lff-interpretations.
\begin{lemma}\label{lemma:spoliers-spoil-entailment-lffs}
Let $\interI$ be a (finite) $(|\queryq|,\names)$-lff-like interpretation and let $\queryq$ be a CQ.
Then $\interI \not\models \queryq$ if there is an $\names$-super-spoiler $\superspoilerKBof{\queryq}$ for \( \queryq \) such that $\interI \models \superspoilerKBof{\queryq}$ and if $\interI \models \superspoilerKBof{\queryq}$ for some $\names$-super-spoiler $\superspoilerKBof{\queryq}$ for $\queryq$ then $\interI \not\models \queryq$.
\end{lemma}
\begin{proof}[Proof (from non-entailment to super-spoilers).]
  We construct a sequence of KBs $\kbK_0 := \emptyset, \kbK_1, \kbK_2, \ldots$ converging to an $\names$-super-spoiler for $\queryq$. 
  To do it, fix some ordering on pairs $(\queryq', \splittingof{\queryq'}{\names})$ of fork rewritings $\queryq'$ and $\names$-splittings of $\queryq'$, and consider $i$-th such pair.
  Observe that $\splittingof{\queryq'}{\names}$ is not compatible with $\interI$. 
  Indeed, otherwise by~\cref{lemma:splittings-linking-notions-together-from-splitting-to-match} we would have $\interI \models \queryq$.
  Thus, there is at least one item of~\cref{def:splittings-compatible} that is not satisfied. 
  Let $\alpha$ be the axiom the corresponding axiom~\cref{def:splittings-compatible}.
  Note that $\interI \not\models \alpha$.
  Hence, let $\beta$ be the corresponding ``negated'' axiom from~\cref{def:spoiler}.
  We put $\kbK_i := \kbK_{i-1}$ if $\beta$ is already in $\kbK_{i-1}$ and $\kbK_i := \kbK_{i-1} \cup \{ \beta \}$ otherwise.
  From the definition of a spoiler, we see that $\kbK_i$ is an $\names$-spoiler for the $i$-th pair. Moreover, $\interI \models \beta$.
  Hence, the last KB on the list is the desired $\names$-super-spoiler $\superspoilerKBof{\queryq}$ for $\queryq$ and, by the construction, $\interI$ is a (finite) model of $\superspoilerKBof{\queryq}$.
\end{proof}
\begin{proof}[Proof (from a super-spoiler to non-entailment).]
  Ad absurdum, assume $\interI \models \queryq$.
  Hence, by~\cref{lemma:splittings-linking-notions-together-from-splitting-to-match} we infer that there is a fork rewriting \( \queryq' \) of~\( \queryq \) and an $\names$-splitting \( \splittingof{\queryq'}{\names} \) of \( \queryq' \) that is compatible with \( \interI \).
  Since \( \superspoilerKBof{\queryq} \) is an $\names$-super-spoiler for $\queryq$ we have that, by~\cref{def:super-spoiler}, it is also an $\names$-spoiler for \( \splittingof{\queryq'}{\names} \). 
  This implies that for \( \splittingof{\queryq'}{\names} \) at least one of the conditions \ref{item:A:def:spoiler-notion}--\ref{item:C:def:spoiler-notion} from~\cref{def:spoiler} hold, contradicting the compatibility of \( \splittingof{\queryq'}{\names} \) with \( \interI \). 
\end{proof}

Thus, relying on the presented lemma we conclude a reduction from the (U)CQ entailment problem to the problem of checking an existence of an $\indK$-super-spoiler spoiling the (finite) satisfiability of $\kbK$.
\begin{lemma}\label{lemma:spoliers-spoil-ucq-over-locally-forward-dls}
Let \( \abstrDL \) be (finitary) locally-forward abstract DL, \( \kbK \) be a (finitely) satisfiable $\abstrDL$-KB and $\queryq = \textstyle\bigvee_{i=1}^{m} \queryq_i$ be a UCQ.
Then \( \kbK \not\modelsoptfin \queryq \) iff there are \( \indK \)-super-spoilers $\superspoilerKBof{\queryq_i}$ for all \( \queryq_i \) s.t. $\kbK \cup \textstyle\bigcup_i \superspoilerKBof{\queryq_i}$ is (finitely) satisfiable.
\end{lemma}
\begin{proof}
  Note that since super-spoiler are $\ALCcap$-KBs, they belong to $\abstrDL$ by definition.
  For the right-to-left direction, assume towards a contradiction that $\kbK \modelsoptfin \queryq$ holds and take any (finite) $(|\queryq|,\indK)$-lff-like model $\interI$ of $\kbK \cup \textstyle\bigcup_i \superspoilerKBof{\queryq_i}$ (guaranteed by~\cref{fact:only-lff-like-countermodels-matters}).
  By assumption we conclude $\interI \models \queryq$ and hence $\interI \models \queryq_i$ for some~$1 \leq i \leq m$.
  But $\interI \models \superspoilerKBof{\queryq_i}$, which contradicts~\cref{lemma:spoliers-spoil-entailment-lffs}. 
  For the other direction, take any (finite) $(|\queryq|,\indK)$-lff-like countermodel $\interI$ for $\kbK$ and $|\queryq|$ (guaranteed by~\cref{fact:only-lff-like-countermodels-matters}).
Hence, for each $1 \leq i \leq n$ we have that $\interI \not\models \queryq_i$ and by~\cref{lemma:spoliers-spoil-entailment-lffs} we get an $\indK$-super-spoiler $\superspoilerKBof{\queryq_i}$ for~$\queryq_i$ such that $\interI \models \superspoilerKBof{\queryq_i}$.
Thus $\interI \models \kbK \cup \textstyle\bigcup_i \superspoilerKBof{\queryq_i}$.
\end{proof}

\subsection{Step V: Super-spoilers made small and efficient}

To get the optimal complexity bounds, we need to show that there are exponentially many super-spoilers that can be enumerated in exponential time and that the size of each super-spoiler is only of polynomial size.

We first show the following lemma (an analogous of~\cite[Lemma 5]{LutzDL08}) proving small size of super-spoilers.

\begin{lemma}\label{lemma:special-shape-of-super-spoilers}
Let $\names \subseteq \Ilang$ be finite, $\queryq$ be a CQ and let $\superspoilerKBof{\queryq}$ be an \( \names \)-super-spoiler for \( \queryq \).
Then all the axioms contained in $\superspoilerKBof{\queryq}$ are of one of the following forms:
    \begin{enumerate}[label={(\Alph*$'$)}]\itemsep0em
    \item \( \topconcept \dlsubseteq \neg \matchConcept{\hat{\queryq}} \) for some forward-tree-shaped query \( \hat{\queryq} \in \QTree(\maximalforkrew{\queryq}) \),\label{item:A:def:spoiler-notion-simpler}
    \item \( \neg \conceptA(\indva) \) for some name \( \indva \in \names \) and a concept name $\conceptA$ occurring in $\maximalforkrew{\queryq}$,\label{item:B:def:spoiler-notion-simpler}
    \item \( \neg \roler(\indva, \indvb) \) for some names \( \indva,\indvb \in \names \) and a role name $\roler$ occurring in $\maximalforkrew{\queryq}$,\label{item:C:def:spoiler-notion-simpler}
    \item \( \left( \neg \exists \left( \roles_1 \cap \roles_2 \cap \ldots \cap \roles_k \right)\matchConcept{\hat{\queryq}} \right) \left( \indva \right) \) 
    for some forward-tree-shaped query \( \hat{\queryq} \in \QTree(\maximalforkrew{\queryq}) \), name \( \indva \in \names \) and role conjunction $\roles_1 \cap \roles_2 \cap \ldots \cap \roles_k$ occurring in $\maximalforkrew{\queryq}$.\label{item:D:def:spoiler-notion-simpler}
    \end{enumerate}
\end{lemma}
\begin{proof}
  Take any \( \names \)-super-spoiler for \( \queryq \) and let $\alpha$ be any of its axioms and we show that are of the shape above.
  If $\alpha$ is of the form of~\cref{item:B:def:spoiler-notion} or~\cref{item:C:def:spoiler-notion} we are done by the fact that (1) if $\roler$/$\conceptA$ occurs in $\queryq$ then it also occurs in the maximal fork rewriting (2) $\splittingname$ assigns values from $\names$.
  If $\alpha$ is of the form of~\cref{item:A:def:spoiler-notion} we invoke~\cref{lemma:query-trees-are-contained-in-max-fork-rew}.
  Applying all mentioned arguments we are also done with the case when $\alpha$ has the form from~\cref{item:D:def:spoiler-notion}.
\end{proof}

As a direct consequence of the above lemma we obtain:
\begin{lemma}\label{lemma:super-spoiler-small-size}
The size of every \( \names \)-super-spoiler for a CQ \( \queryq \) is polynomial in $|\queryq|+|\names|$ and the total number of \( \names \)-super-spoilers is exponential in $|\queryq|+|\names|$. 
\end{lemma}
\begin{proof}
Let $\maximalforkrew{\queryq} = \queryq^*$.
To bound the size of \( \names \)-super-spoilers we invoke~\cref{lemma:special-shape-of-super-spoilers} and see that the axioms of the corresponding items can be bounded, respectively, by $|\QTree(\queryq^*)|$, $|\queryq|\cdot|\names|$, $|\queryq|\cdot|\names|^2$ and $|\queryq|\cdot|\QTree(\queryq^*)|\cdot|\names|$.
Since $|\QTree(\queryq^*)|$ is bounded polynomially in $|\queryq|$ we are done.
The latter part is now immediate. 
\end{proof}

The last property in our path leading to an algorithm solving the (U)CQ-entailment is the ability to enumerate $\names$-super-spoilers in exponential~time.

\begin{lemma} \label{lemma:enumerate-spoilers}
The set of all $\names$-super-spoilers for a CQ $\queryq$ can be enumerated in time exponential in $|\names|+|\queryq|$.
\end{lemma}
\begin{proof}
We enumerate $\names$-super-spoilers as follows. 
We first enumerate all $\ALCcap$-KBs containing only the axioms stated in~\cref{lemma:special-shape-of-super-spoilers} (requires time exponential in $|\names|+|\queryq$).
To check if a knowledge-base is indeed an $\names$-super-spoiler, we go through all fork rewritings (there are exponentially many in $|\queryq|$ of them) and all splittings for them (exponential in $|\names|+|\queryq|$). 
Then we apply the definition of $\names$-spoilers to check if the considered knowledge-base indeed blocks the splitting, which can be performed, after fixing an $\names$-spoiler and an $\names$-splitting, in polynomial time. 
The execution times are multiplied, hence the overall algorithm works in time exponential in $|\names|+|\queryq|$.
\end{proof}

\subsection{Step VI: The algorithm}

We are ready to present an algorithm for deciding (finite) (U)CQ entailment problem over (finitary) locally-forward DLs, that is worst-case optimal in many scenarios, \eg in the case when the (finite)satisfiability problem for the DL is $\ExpTime$-complete.
We present a pseudocode below.

\begin{algorithm}[h]\label{algorithm:ucq-entailment}
  \DontPrintSemicolon
  \KwData{A UCQ $\queryq = \textstyle \bigvee_{i=1}^{m} \queryq_i$ and an $\abstrDL$-KB $\kbK$.}
  \caption{Checking (finite) UCQ entailment over (finitary) locally-forward abstract DL KBs}\label{algo:ucq-entailment-over-L}

  If $\kbK$ is not (finitely) satisfiable \textbf{return} \texttt{True}. \tcp*{Checkable in $\SAT_{\abstrDL}(\textsf{poly}(\kbK))$.}
  \ForEach{selection of $\indK$-super spoilers $\superspoilerKBof{\queryq_1}, \ldots, \superspoilerKBof{\queryq_m}$ for $\queryq_1, \ldots, \queryq_m$ \tcp*{In $\mathsf{exp}(|\indK|{+}|\queryq|)$ by~L.~\ref{lemma:enumerate-spoilers} } }{%
    If $\kbK \cup \textstyle\bigcup_i \superspoilerKBof{\queryq_i}$ is (finitely) satisfiable \textbf{return} \texttt{False}.\tcp*{In $\SAT_{\abstrDL}(\mathsf{poly}(|\kbK| + |\queryq|))$ by~\cref{lemma:super-spoiler-small-size}}
  }
  \textbf{return} \texttt{True}.
\end{algorithm}

\begin{lemma}\label{lemma:correctness}
Procedure~\ref{algo:ucq-entailment-over-L} returns \texttt{True} iff $\kbK \modelsoptfin \queryq$. 
Moreover, Procedure~\ref{algo:ucq-entailment-over-L} can be implemented to work in time $\mathsf{exp}(|\kbK|+|\queryq|) \cdot \SAT_{\abstrDL}(\mathsf{poly}(|\kbK| {+} |\queryq|))$ for some polynomial function $\mathsf{poly}$ and an exponential function $\mathsf{exp}$ and with $\SAT_{\abstrDL}$ denoting the worst-case optimal running time of the (finite) satisfiability problem for $\abstrDL$-KBs.
\end{lemma}
\begin{proof}
For the first statement of the lemma we consider the following cases.
If $\kbK$ is not (finitely) satisfiable then it entails every query. Our procedure returns $\texttt{True}$ in this case.
If $\kbK$ is (finitely) satisfiable but does not (finitely) entail $\queryq$, then by~\cref{lemma:spoliers-spoil-ucq-over-locally-forward-dls} there are $\indK$-super-spoilers $\superspoilerKBof{\queryq_i}$ for $\queryq_i$ such that $\kbK \cup \textstyle\bigcup_i \superspoilerKBof{\queryq_i}$ is (finitely) satisfiable and hence, the fourth line of the algorithm returns $\texttt{False}$.
Otherwise, $\kbK$ is (finite) satisfiable and (finitely) entails~$\queryq$. Thus again, by~\cref{lemma:spoliers-spoil-ucq-over-locally-forward-dls}, there are no such $\indK$-super-spoilers and so the (finite) satisfiability test in the 4th line of Procedure~\ref{algo:ucq-entailment-over-L} will never succeed. 
Hence, the 5th line will be executed, returning~$\texttt{True}$.\\ 
The second part of the lemma follows immediately from of~\cref{lemma:super-spoiler-small-size} and~\cref{lemma:enumerate-spoilers} and from the fact that $\SAT_{\abstrDL}(\textsf{poly}(\kbK)) {+} \mathsf{exp}(|\indK|{+}|\queryq|) \cdot \SAT_{\abstrDL}(\mathsf{poly}(|\kbK| {+} |\queryq|))$ is bounded by $\mathsf{exp}(|\kbK|{+}|\queryq|) \cdot \SAT_{\abstrDL}(\mathsf{poly}(|\kbK| {+} |\queryq|))$.
\end{proof}

Relying on the above lemma, we conclude our main theorem. 
\begin{thm}
For any (finitary) locally-forward abstract DL $\abstrDL$ with (finite) $\abstrDL$-KB-satisfiability problem decidable in time $\SAT_{\abstrDL}(\cdot)$, there exists a polynomial and an exponential function $\mathsf{poly}$ and $\mathsf{exp}$ such that the (finite) UCQ-entailment problem over $\abstrDL$-KB is decidable, for an input $\kbK, \queryq$, in time $\mathsf{exp}(|\kbK|+|\queryq|) \cdot \SAT_{\abstrDL}(\mathsf{poly}(|\kbK| {+} |\queryq|))$.
\end{thm}

The most important application of our work is when the (finite) knowledge base satisfiability problem for $\abstrDL$ is $\ExpTime$-complete.
Then the function $\mathsf{exp}(|\kbK|+|\queryq|) \cdot \SAT_{\abstrDL}(\mathsf{poly}(|\kbK| {+} |\queryq|))$ is actually a single exponential function, and hence, we have the following corollary (the lower bound follows from $\ALC$~\cite[Thm. 5.13]{dlbook}).
\begin{corollary}
The (finite) (U)CQ entailment problem is $\ExpTime$-complete for any (finitary) locally-forward abstract DL with $\ExpTime$-complete (finite) knowledge base satisfiability problem.
\end{corollary}

Recall from the beginning of the section that $\ALCSCC$ is finitary locally-forward and that any abstract DL $\abstrDL$ contained in $\ALCHbregQ$ is locally-forward. Since their corresponding satisfiability problems are $\ExpTime$-complete, by the above corollary we conclude:
\begin{thm}
The finite UCQ entailment problem for $\ALCSCC$ is $\ExpTime$-complete and the UCQ entailment problem for any $\ALC \subseteq \abstrDL \subseteq \ALCHbregQ$ is $\ExpTime$-complete.
\end{thm}
This closes numerous gaps in the complexity of query entailment that were present in the literature, \eg the complexities of CQ entailment for $\ALC\textit{b}$ or UCQ entailment for $\ALCHQ$ were unknown. It also proves that regular role expressions from $\ALCHbregQ$ do not increase the complexity of querying, as it is the case for nominals, inverses or self-loops.
As a last remark: any PEQ can be transformed into a (U)CQ of exponential size. 
This yields us $\TwoExpTime$-completeness of PEQ querying for any logics mentioned in the above theorem. The lower bound holds already for $\ALC$~\cite[Thm.~1]{OrtizS14}.

\section*{Acknowledgements}
I thank my supervisor, Sebastian Rudolph, for many useful comments and, most notably, for his extreme patience while supervising me.
This work is supported by the ERC Consolidator Grant No.~771779 (\href{https://iccl.inf.tu-dresden.de/web/DeciGUT/en}{DeciGUT}).
\begin{figure}[H]
    \centering
    \includegraphics[scale=0.045]{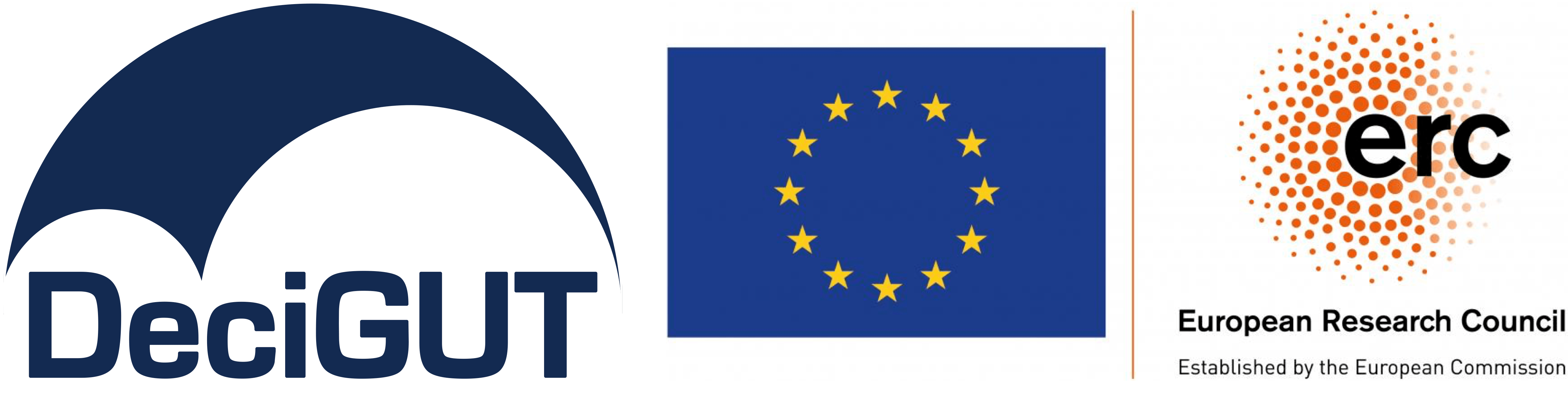}
\end{figure}
\bibliographystyle{alpha}
\bibliography{references}
\end{document}